\documentclass[journal,comsoc]{IEEEtran}

\usepackage{amsmath,amsfonts, amsthm, amssymb}
\usepackage{array}
\usepackage[caption=false,font=normalsize,labelfont=sf,textfont=sf]{subfig}
\usepackage{textcomp}
\usepackage{stfloats}
\usepackage{url}
\usepackage{verbatim}
\usepackage{graphicx}
\usepackage{cite}
\usepackage{xcolor}
\usepackage{bbm}
\usepackage{algorithm}
\usepackage{algpseudocode}
\usepackage{booktabs}
\theoremstyle{definition}
\newtheorem{thm}{Theorem}[section]
\newtheorem{defn}[thm]{Definition}
\newtheorem{lem}[thm]{Lemma}

\newtheorem{prop}[thm]{Proposition}

\newtheorem{rem}[thm]{Remark}

\newcommand{\argmax}{\operatornamewithlimits{arg\,max}}

\newcommand{\IASPA}[1]{\operatorname{IA\text{-}SPA}(#1)}

\newcommand{\lukas}[1]{\textbf{\textcolor{red}{[Lukas: #1]}}}

\begin{document}

\title{Optimal Transmitter Placement in Realistic Urban Environments}

\author{Lukas~Taus,
        Richard~Tsai, and~Jeffrey~G.~Andrews,~\IEEEmembership{Fellow,~IEEE}
\thanks{L. Taus and R. Tsai are with the Oden Institute for Computational Engineering and Sciences, The University of Texas at Austin, Austin, TX, USA (e-mail: l.taus@utexas.edu; ytsai@math.utexas.edu). Taus and Tsai are partially supported by NSF grant DMS-2208504, DMS-2513857, and Army Research Oﬃce Grant W911NF2320240.}%
\thanks{J. G. Andrews is with 6G@UT in the Wireless Networking and Communications Group and the Dept. of Electrical and Computer Engineering at the University of Texas at Austin, Austin, TX, USA (e-mail: jandrews@ece.utexas.edu).}%
\thanks{Manuscript last updated \today.}
\thanks{This work has been submitted to the IEEE for possible publication. Copyright may be transferred without notice, after which this version may no longer be accessible.}}



\maketitle

\begin{abstract}
In a wireless network, the spatial location of transmitters has a large impact on the achievable rate at each user location. Determining where to deploy transmitters—for example, cellular base stations—is a difficult non-convex problem, typically addressed using simplified propagation models and heuristics that only partially capture site-specific geometry and user demand. We propose a mathematically rigorous framework that instead incorporates detailed site-specific maps, spatial material properties, and realistic signal attenuation. We introduce an aggregated network-quality functional that scores receiver-weighted signal quality, impose deployment cost through cardinality and budget constraints, and establish submodularity of the resulting objective under practical conditions. To solve the resulting resource-constrained optimization problem for sparse, discrete transmitter configurations, we propose the Interference-Aware Submodular Placement Algorithm (IA-SPA) and establish a theoretical approximation guarantee relative to the global optimum of the surrogate objective. IA-SPA can incorporate existing base-station locations and prohibited deployment areas, making it applicable to both clean-slate and incremental deployments. We evaluate our approach using a ray-tracing-based simulation framework applied to 3D maps of San Francisco and Florence, and compare against known base-station deployments by AT\&T, T-Mobile, and Iliad. The proposed placement strategy achieves significant increases in mean data rate (about 2x) and edge rate (2–8x) compared with existing base-station deployments using the same number of transmitters. These gains persist under multiple simultaneous exclusionary zones, small-scale fading, perturbations to the material and geometric properties of the environment, and incremental densification of already substantial existing networks. The resulting pipeline has linear computational dependence on the
number of candidate sites for fixed receiver discretization and
deployment size, while wall-clock measurements demonstrate practical
runtime at city scale.

\end{abstract}

\begin{IEEEkeywords}
Base station placement, Site-Specific Optimization, Combinatorial Optimization, Ray Tracing, Digital Twins.
\end{IEEEkeywords}

\section{Introduction}

In the 6G era, cellular networks will continue to slowly densify and support new spectral bands, while also becoming more amenable to data-driven optimization.    Given strong constraints on network cost and energy consumption~\cite{AndHum25}, the placement of each new base station should attempt to maximize the coverage and capacity enhancement it provides to the network.   The rapidly improving technologies of network digital twins, based on increasingly sophisticated site-specific mapping and ray tracing, provide a promising pathway to improved network deployment and configuration~\cite{yildiz2025digital,AhmDigTwin25}. For these emerging tools to be able to enable significant increases in cell edge rate and median rate through improved network deployment, they need to be harnessed by well-designed, flexible and principled optimization algorithms that are able to leverage their capabilities and approach global optima.

However, optimizing something as seemingly straightforward as base station (BS) placement is in fact very complicated.  The optimal BS locations depend on not only the site-specific topology and signal propagation, but also on the user locations and demand probability density functions, and other constraints such as prohibited regions or respecting existing BS deployments.   Furthermore, while each new BS increases signal quality in its immediate vicinity, it causes new interference to all the neighboring cells, and changes the association regions (cell boundaries).  Thus, the performance of users in adjacent cells are strongly coupled to one another, which renders the globally optimal placement of transmitters non-convex and computationally intractable.  The goal of the present paper is to provide a general and rigorous approach to BS placement, one that is able to ingest detailed site-specific information and propagation modeling, but is also computationally feasible.  We focus on the downlink of a cellular network, but the framework is broadly applicable to the placement of any set of radiating transmitters.  

\subsection{Relation to Prior Work}\label{sec:related}

When the signal frequencies are sufficiently high relative to the domain features, wireless transmitter placement is, to first order, a geometric problem. Foundational results from the \emph{art gallery problem}~\cite{orourke1987art, bjorling1995, ghosh2008} establish bounds for covering polygonal environments, serving as a proxy for signal presence in high-frequency or aerial scenarios~\cite{cho2025optimal}. However, since finding optimal configurations is NP-hard~\cite{lee1986computational}, practical scalability requires greedy strategies~\cite{krause2011submodular}. Early efforts by Amaldi et al.~\cite{1230131} used randomized heuristics for UMTS planning, but relied on simplified propagation models and constant interference thresholds that lacked uniqueness. This critique extends to classical facility-location and coverage-optimization formulations more broadly (e.g., maximal covering location problems, set-covering location, and integer-programming variants~\cite{mathar2000integer}): they share the same reliance on simplified, distance- or binary-threshold-based propagation and cannot natively accommodate ray-traced, material-aware signal attenuation or the submodular performance guarantee we establish in Section~\ref{sec:systemmodel}.



In both robotics and communications, greedy algorithms are often used for managing the complexity of coverage problems. In robotics, \emph{next-best-view} and exploration algorithms are used to optimize gain measures like unexplored volume or surface area~\cite{bircher2018receding,gonza2002navigation, 8794426, popovic2021volumetric}. In addition, multi-coverage objectives have been explored to make the sensor network more robust against failure~\cite{taus2024}. However, a critical gap is the lack of detailed analysis on the role of interference. Most existing greedy frameworks optimize for simple area coverage or connectivity, but particularly in urban networks, interference plays a critical role near the cell edge in degrading the signal quality, i.e. the Signal-to-Interference plus Noise Ratio (SINR). 


Physically grounded optimization is increasingly enabled by digital twins and ray-tracing frameworks like Sionna RT~\cite{hoydis2022sionna}, which incorporate site-specific material properties and antenna orientations~\cite{aram2024sitespecificoutdoorpropagationassessment, yildiz2025digital}. Such high-fidelity modeling is critical in mmWave and Reconfigurable Intelligent Surface (RIS) networks~\cite{uwaechia2020, renzo2020ris, 10118599, 11309718, zhi2024ris, an2024hapswipt}, including in satellite and space-air-ground integrated systems~\cite{lin2025sagin, he2022satellite}. As a data-efficient, lower-complexity alternative to full ray tracing, scatterer-model-driven channel gain maps have also been used to guide network topology and RIS deployment decisions~\cite{sun2026irs, sun2026cgm}; our framework is agnostic to the specific choice of propagation model and can equally ingest such map-based approximations in place of $P(y,x)$ when ray-traced fields are unavailable or too costly to generate.
Gradient-based methods have also been explored for transmitter placement~\cite{Belgiovine_2026}. Such approaches optimize a continuous, generally non-convex objective and can depend on initialization; they also typically require the deployment size to be specified in advance. Direct SINR optimization remains non-convex~\cite{qian2008mapelachievingglobaloptimality, 4275017, 1626432, 1019292, 4469894}, so local optimization does not in general provide a global performance guarantee. In contrast, our formulation optimizes over a discrete candidate set using a submodular surrogate objective. This yields an explicit approximation guarantee for the greedy placement rule, while the threshold-constrained formulation can determine the number of added sites adaptively.


To bypass the computational cost of high-fidelity simulation, various acceleration techniques have been proposed, ranging from generative AI~\cite{chen2021fastrack, khoramnejad2024generativeaioptimizationnextgeneration} and neural surrogates~\cite{pasqual2024optimized} to DRL-based placement for aerial and RIS-assisted networks~\cite{9149258, 9700536, 10847914}. Other strategies employ Bayesian optimization~\cite{10697138}, genetic algorithms~\cite{10.1007/s42979-022-01533-y}, or multi-armed bandits~\cite{erden2020outdoormmwavebasestation} for transmission and site design. Additionally, terrain-aware frameworks utilize radio tomographic imaging~\cite{das2024skyscale, 9746987, 10486853} or global connectivity maps~\cite{11310308, 10373821} to optimize real-time node movement in site-specific channels.

The contribution of this paper lies in combining submodular greedy optimization with site-specific propagation information. Unlike iterative parameter tuning~\cite{10254476} or DRL-based deployment under EMF and coordination constraints~\cite{mallik2025basestationdeploymentemf, 11311373}, our framework focuses on discrete site selection and separates propagation-map generation from the combinatorial placement step. This structure permits detailed ray-traced fields to be reused throughout the optimization and provides a formal approximation guarantee for the surrogate network-quality objective.

\subsection{Contributions}\label{sec:contributions}

We formulate transmitter placement using site-specific propagation fields and a receiver-weighted network-quality objective. The functional $\mathcal{S}(T;\mathcal{P})$ measures utility derived from aggregated received power under a prescribed spatial demand density, while deployment cost is imposed separately through cardinality or resource constraints.

We then develop the Interference-Aware Submodular Placement Algorithm (IA-SPA), a greedy method for constructing sparse, discrete transmitter configurations. The principal contributions are as follows.

\textbf{Optimization Framework}: We formulate resource-constrained transmitter placement directly over a discrete set of feasible sites while retaining detailed propagation information generated by a site-specific model. The resulting greedy algorithm has a provable approximation guarantee for the surrogate objective.

\textbf{Network-Quality Functional}: We introduce $\mathcal{S}(T;\mathcal{P})$ and analyze it for two aggregation surrogates, $\mathcal{P}_{\mathrm{MAX}}$ and $\mathcal{P}_{\mathrm{SUM}}$. Under a concave non-decreasing utility map, the resulting set function is monotone and submodular, which provides the diminishing-returns structure used by IA-SPA.

\textbf{Site-Specific Validation}: We evaluate the method using Sionna RT ray tracing in 3D models of San Francisco and Florence and compare against reference deployments from AT\&T, T-Mobile, and Iliad. In the San Francisco AT\&T comparison, for example, $\IASPA{\mathcal{P}_{\mathrm{MAX}}}$ increases mean data rate by about $70\%$ and edge rate by about $88\%$ while slightly reducing mean interference.

\textbf{Constrained and Incremental Deployments}: The framework permits arbitrary feasible candidate sets, including multiple non-contiguous exclusion zones, and can optimize new sites while holding an existing deployment fixed. This allows the same algorithmic formulation to address both clean-slate placement and incremental network expansion.

\textbf{Robustness and Scalability}: Monte Carlo and environment-perturbation experiments show that the observed placement advantage persists under the small-scale fading and material, clutter, and geometric perturbations considered here. For fixed receiver discretization and deployment size, the greedy stage is linear in the number of candidate sites, while radio-map generation is embarrassingly parallel across candidates. Direct wall-clock measurements demonstrate practical runtime at the city scale considered in this study.

\subsection{Notation and Organization}\label{sec:notation}
Throughout this paper, we use the following mathematical notation:
\begin{itemize}
    \item $\Omega \subseteq \mathbb{R}^n$: A compact set representing the physical environment.
    \item $x, y \in \Omega$: Locations within the domain where $y$ typically denotes a receiver position and $x$ a  transmitter position.
    \item $\mu(x)$: Material parameters characterizing local signal propagation (e.g., absorption, reflection).
    \item $P(y, x)$: Received power at $y$ from a transmitter at $x$; in the numerical experiments it is expressed in noise-normalized units.
    \item $T = \{t_1, \dots, t_n\}$: A discrete set of transmitter locations.
    \item $\mathcal{P}_\text{SUM}(y, T)$ and $\mathcal{P}_\text{MAX}(y, T)$: Two aggregation surrogates constructed from the single-transmitter fields $P(y,x)$.
    \item $f(y)$: A normalized non-negative spatial density used to weight receiver locations; equivalently, it may be interpreted as a probability density for user location.
    \item $w(\kappa)$: A weight function representing the marginal utility of network quality $\kappa$.
    \item $\bar{W}(x) = \int_0^x w(\kappa) d\kappa$: The anti-derivative of $w$.
    \item $\mathcal{S}(T;\mathcal{P})$: The aggregated network quality functional for aggregation measure $\mathcal{P}$; when $\mathcal{P}$ is clear from context, we write $\mathcal{S}(T)$.
    \item $\mathcal{G}_{\mathcal{P}}(x\mid T)$: The marginal gain from adding candidate $x$ to $T$ under aggregation measure $\mathcal{P}$.
    \item $\IASPA{\mathcal{P}}$: IA-SPA instantiated with aggregation measure $\mathcal{P}$.
\end{itemize}

\textbf{Organization:} Section~\ref{sec:systemmodel} establishes the system model, detailing the propagation environment, signal aggregation metrics, and the formal optimization problem. Section~\ref{sec:alg} introduces the proposed IA-SPA framework, alongside an analysis of its theoretical performance guarantees and computational complexity. Section~\ref{sec:num_results} presents a comprehensive numerical evaluation using high-fidelity 3D ray tracing in San Francisco and Florence, exploring diverse deployment scenarios including exclusionary zones and incremental network expansion. Finally, Section~\ref{sec:conclusion} discusses the broader implications of our findings and provides concluding remarks.

\section{System Model and Problem Formulation}\label{sec:systemmodel}

Let $\Omega \subseteq \mathbb{R}^n$ be a compact set representing the environment. We allow spatially varying material parameters $\mu(x)$, which determine how signals interact with buildings, terrain, and other objects through effects such as absorption, reflection, refraction, and transmission.

To illustrate the formulation, we use the simple urban scene shown in Fig.~\ref{fig:toy_map}. The numerical experiments in Section~\ref{sec:num_results} use the substantially more detailed San Francisco and Florence scenes shown in Fig.~\ref{fig:sf_map}.

\begin{figure}
    \centering
    \includegraphics[width=0.32\linewidth]{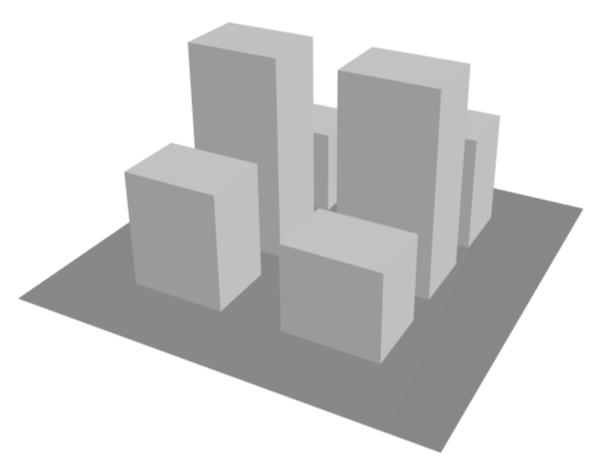}
    \includegraphics[width=0.20\linewidth]{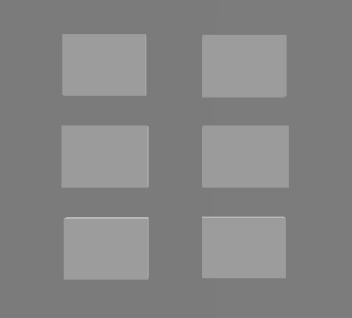}
    \includegraphics[width=0.20\linewidth]{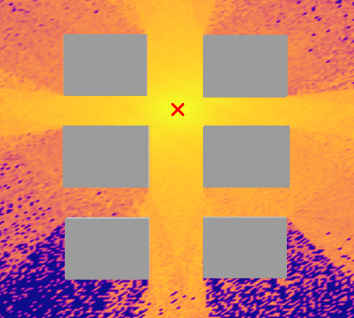}
    \caption{Toy urban scene used to illustrate the formulation. Left:
    3D view. Center: top-down view. Right: achievable-rate field
    induced by a single transmitter (red marker), computed from the
    site-specific propagation field $P(y,x)$.}
    \label{fig:toy_map}
\end{figure}

\subsection{Propagation Model}
We define the propagation function $P(y,x)$ as the power received at location $y\in\Omega$ from a transmitter at location $x\in\Omega$. The function depends on the geometry of the entire environment and on its spatially varying material parameters $\mu$. It therefore encodes attenuation and the effects of interactions with buildings, terrain, and other objects. Once the noise power and interference are specified, $P(y,x)$ determines the SNR, SINR, and ultimately the achievable data rate at $y$.

The appropriate model for $P(y,x)$ depends on the desired balance between computational cost and physical fidelity. Three common paradigms, arranged here in order of increasing physical realism, are:

\begin{enumerate}
    \item \textbf{Line-of-sight (LoS) visibility models.} These geometric models classify a transmitter--receiver pair according to whether the direct path is obstructed, often producing binary covered/occluded regions or augmenting this classification with a distance-based path-loss law~\cite{cho2025optimal}. They are inexpensive, but do not explicitly resolve energy arriving along indirect paths.

    \item \textbf{Ray-tracing models.} These high-frequency approximations trace multiple propagation paths through a site-specific scene and account for interactions such as reflection, diffraction, transmission, and scattering at material boundaries. They capture multipath propagation and blockage-induced attenuation while remaining substantially less expensive than a full-wave solution.

    \item \textbf{Full-wave models.} These models numerically solve Maxwell's equations, resolving wave phenomena without the high-frequency approximations underlying geometric ray tracing. They provide the greatest physical fidelity among the three paradigms, but their computational and memory costs generally limit their use in large urban placement studies.
\end{enumerate}

Our analysis is agnostic to which of these paradigms, or which empirical or learned alternative, is used to construct $P(y,x)$. For the numerical experiments in this paper, we compute $P(y,x)$ using the site-specific Sionna RT ray tracer~\cite{hoydis2022sionna}. We normalize the noise power to $\sigma^2=1$ and express $P(y,x)$ in noise-normalized units; consequently, $P(y,x)$ is numerically equal to the single-link SNR. This normalization does not affect the placement formulation.

The right panel of Fig.~\ref{fig:toy_map} illustrates the spatial structure induced by one such propagation field. The irregular rate pattern reflects site geometry and material interactions that cannot be represented by distance or binary visibility alone.

\subsection{Network Quality and SINR}
For a transmitter set $T=\{t_1,\ldots,t_n\}$, we assume that each receiver associates with the transmitter providing the largest received power. The resulting SINR is
\begin{equation}
    \mathrm{SINR}(y,T)
    =
    \frac{\max_{t\in T}P(y,t)}
    {\sum_{t\in T}P(y,t)-\max_{t\in T}P(y,t)+\sigma^2}.
\end{equation}
Thus, the strongest received signal contributes useful power, while all other active transmitters contribute interference. Directly optimizing this SINR over transmitter locations leads to a difficult non-convex subset-selection problem~\cite{qian2008mapelachievingglobaloptimality, 4275017, 1626432, 1019292, 4469894}. We therefore introduce two tractable aggregation surrogates. The final deployments are nevertheless evaluated using the physical SINR above, as described in Section~\ref{sec:num_results}.

\begin{defn}[Aggregated Network Quality]\label{def:agg_quality}
Given the single-transmitter propagation field $P(y,x)$ and a transmitter set $T$, define
\begin{enumerate}
    \item \textbf{Maximum received-power surrogate} ($\mathcal{P}_{\mathrm{MAX}}$):
    \begin{equation}
        \mathcal{P}_{\mathrm{MAX}}(y,T)=\max_{x\in T}P(y,x).
    \end{equation}
    This surrogate retains only the strongest received signal. Under the association model above, it therefore coincides with the useful-signal term in the SINR numerator, while ignoring interference.

    \item \textbf{Total received-power surrogate} ($\mathcal{P}_{\mathrm{SUM}}$):
    \begin{equation}
        \mathcal{P}_{\mathrm{SUM}}(y,T)=\sum_{x\in T}P(y,x).
    \end{equation}
    This surrogate rewards all received power, without distinguishing the serving transmitter from interferers. It is therefore an optimization proxy rather than an upper bound on the achievable data rate.
\end{enumerate}
\end{defn}

The two aggregation measures differ fundamentally in how they value overlapping coverage. Under $\mathcal{P}_{\mathrm{MAX}}$, adding a transmitter improves the network quality at a receiver only if it provides a stronger signal than the best transmitter already serving that receiver. In contrast, $\mathcal{P}_{\mathrm{SUM}}$ assigns positive value to any additional received power, even in regions that are already well served by other transmitters. Thus, $\mathcal{P}_{\mathrm{MAX}}$ naturally emphasizes improvements to the strongest serving link, whereas $\mathcal{P}_{\mathrm{SUM}}$ emphasizes total received power. As shown in Section~\ref{sec:alg}, this distinction is also reflected mathematically: $\mathcal{P}_{\mathrm{MAX}}$ is monotone submodular, while $\mathcal{P}_{\mathrm{SUM}}$ is modular.

\subsection{Accommodating Deployment Constraints}

In practical network planning, environmental, regulatory, or infrastructural constraints may prohibit deployment over part of the domain. We represent these restrictions directly through the feasible candidate set $\mathcal{X}\subset\Omega$ and require every newly selected transmitter to lie in $\mathcal{X}$. No modification of the propagation field $P(y,x)$ is needed: an infeasible site is simply excluded from the optimization.

Furthermore, the algorithm naturally supports \textit{incremental deployment}. Let $T_{\mathrm{fixed}}$ denote a set of pre-existing transmitters whose locations remain fixed. If at most $k$ new transmitters may be added, the optimization problem becomes
\begin{align}
    \max_{T\subseteq\mathcal{X},\,|T|\leq k}
    \mathcal{S}(T\cup T_{\mathrm{fixed}};\mathcal{P}).
\end{align}
Because fixing an initial set preserves monotonicity and submodularity of the marginal objective, $\IASPA{\mathcal{P}}$ can optimize the added sites while accounting for the propagation footprint of the existing deployment. The same formulation can also enforce mandatory pre-selected sites.

\subsection{Optimization Problem Formulation}\label{sec:optform}

Our goal is to select a small set of transmitter locations that provides high network quality over the service domain $\Omega$.

In practice, demand is spatially heterogeneous. We therefore introduce a non-negative priority density $f:\Omega\to\mathbb{R}_+$ and normalize it so that
\begin{align}
    \int_\Omega f(y)\,dy = 1.
\end{align}
Thus, $f$ can be interpreted as the probability density of user location, while still allowing higher weight to be assigned to high-demand regions such as public squares, transit hubs, or office complexes. A non-normalized priority field can be rescaled to this form; this does not change the ordering of transmitter configurations and only rescales the utility threshold in problem~\eqref{eq:min_t_st_cov}. Throughout, $f$ is treated as a fixed demand profile; adapting transmit powers or scheduling to instantaneous traffic variation is outside the scope of this paper.

We next introduce a non-negative weight function $w(\kappa)$ and its antiderivative
\begin{align}
    \bar W(x)=\int_0^x w(\kappa)\,d\kappa.
\end{align}
The value $w(\kappa)$ represents the marginal utility assigned to an increase in aggregated network quality at level $\kappa$. For each aggregation measure considered below, we assume
\begin{align}\label{eq:M_bound}
    M:=\sup_{y\in\Omega}\mathcal{P}(y,\mathcal{X})<\infty
\end{align}
and $w\in L^1(0,M)$. Since $\mathcal{P}(y,T)\leq \mathcal{P}(y,\mathcal{X})$ for $T\subseteq\mathcal{X}$,
\begin{align}
    0\leq \bar W(\mathcal{P}(y,T))
    \leq \int_0^M w(\kappa)\,d\kappa<\infty.
\end{align}

\begin{rem}
Examples satisfying the assumptions used below include
\begin{align}
    \bar W(x)=\log(1+x)
    \qquad\text{and}\qquad
    \bar W(x)=\frac{x}{x+c},\quad c>0.
\end{align}
\end{rem}

Taking all these considerations into account, we can formulate a function which measures the value of a set of transmitters $T$ using the function

\begin{align} \label{eq:S_def}
    \mathcal{S}(T;\mathcal{P}) = \int_0^M w(\kappa) \int_{\Omega}\mathbbm{1}_{\{\mathcal{P}(y, T) > \kappa\}} f(y) dy d\kappa.
\end{align}
We write the dependence on the aggregation measure $\mathcal{P}$ explicitly
when it is relevant; when $\mathcal{P}$ is clear from context, we use the
shorthand $\mathcal{S}(T)$.

For each fixed admissible aggregation measure $\mathcal{P}$, the functional is monotone:
\begin{align}
    T_1\subseteq T_2
    \implies
    \mathcal{S}(T_1;\mathcal{P})
    \leq
    \mathcal{S}(T_2;\mathcal{P}).
\end{align}
This follows from monotonicity of the two aggregation measures
$\mathcal{P}_{\mathrm{MAX}}$ and $\mathcal{P}_{\mathrm{SUM}}$ introduced in Definition~\ref{def:agg_quality}; throughout the paper, the theoretical claims are restricted to these two choices.

Two natural resource-constrained formulations are then
\begin{align}\label{eq:min_t_st_cov}
  \min_{T\subseteq\mathcal{X}} |T|
  \quad \text{s.t.}\quad
  \mathcal{S}(T;\mathcal{P})\geq\beta,
\end{align}
which minimizes the number of transmitters required to reach a prescribed utility level, and
\begin{align}\label{eq:max_cov_st_t}
  \max_{T\subseteq\mathcal{X}} \mathcal{S}(T;\mathcal{P})
  \quad \text{s.t.}\quad
  |T|\leq\beta,
\end{align}
which maximizes utility under a cardinality budget.

\begin{thm}\label{thm:S_formula}
Let $X$ be a random variable on $\Omega$ with density $f$. Then
\begin{align}
    \mathcal{S}(T;\mathcal{P})
    =
    \mathbb{E}_X\!\left[
        \bar W\!\left(\mathcal{P}(X,T)\right)
    \right].
\end{align}
\end{thm}

\begin{proof}[Proof of Theorem~\ref{thm:S_formula}]
Since $f,w\geq 0$, $\int_{\Omega} f(y)\,dy=1$, and
\[
    \int_0^M w(\kappa)\,d\kappa < \infty,
\]
the integrand is absolutely integrable. By Fubini's theorem,
\begin{align}
\mathcal{S}(T;\mathcal{P})
&=
\int_\Omega f(y)
\int_0^M
w(\kappa)\,
\mathbbm{1}_{\{\mathcal{P}(y,T)>\kappa\}}
\,d\kappa\,dy \\
&=
\int_\Omega f(y)
\int_0^{\mathcal{P}(y,T)}
w(\kappa)\,d\kappa\,dy \\
&=
\int_\Omega
\bar W\!\left(\mathcal{P}(y,T)\right)f(y)\,dy \\
&=
\mathbb{E}_X\!\left[
\bar W\!\left(\mathcal{P}(X,T)\right)
\right].
\end{align}
\end{proof}

The two constrained problems in equations \eqref{eq:min_t_st_cov} and \eqref{eq:max_cov_st_t} are, in general, combinatorial (subset-selection) optimization problems and hence NP-hard~\cite{lee1986computational}. We now establish that $\mathcal{S}(\cdot;\mathcal{P})$ possesses a \emph{submodularity} (diminishing-returns) property. This is the key structural fact behind (i) the approximation guarantee we prove for the greedy, discrete IA-SPA in Section~\ref{sec:alg}, which holds regardless of any convexity structure, and (ii) a complementary convexity result for a continuous relaxation of the same problems, discussed in Remark~\ref{rem:usefulness}.

\begin{defn}[Submodularity]\label{def:submodular}
Let $\mathcal{X}$ be a finite ground set. A set function
$F:2^{\mathcal{X}}\to\mathbb{R}$ is \emph{submodular} if, for every
$A\subseteq B\subseteq\mathcal{X}$ and every
$t\in\mathcal{X}\setminus B$,
\begin{align}
    F(A\cup\{t\})-F(A)
    \geq
    F(B\cup\{t\})-F(B).
    \label{eq:submodular_def}
\end{align}
It is \emph{modular} if equality holds for all such $A,B,t$.
\end{defn}

\begin{lem}\label{lem:P_submodular}
For every $y\in\Omega$, with the convention
$\mathcal{P}_{\mathrm{MAX}}(y,\emptyset)=0$, the set functions
$T\mapsto\mathcal{P}_{\mathrm{SUM}}(y,T)$ and
$T\mapsto\mathcal{P}_{\mathrm{MAX}}(y,T)$ are monotone non-decreasing
on $2^{\mathcal{X}}$. Moreover,
$\mathcal{P}_{\mathrm{SUM}}(y,\cdot)$ is modular and
$\mathcal{P}_{\mathrm{MAX}}(y,\cdot)$ is submodular.
\end{lem}

\begin{proof}[Proof of Lemma~\ref{lem:P_submodular}]
Fix $A\subseteq B\subseteq\mathcal{X}$ and
$t\in\mathcal{X}\setminus B$. For the sum surrogate,
\begin{align}
\mathcal{P}_{\mathrm{SUM}}(y,B)-\mathcal{P}_{\mathrm{SUM}}(y,A)
=
\sum_{x\in B\setminus A}P(y,x)\geq0,
\end{align}
and the marginal gain of adding $t$ is
\begin{align}
\mathcal{P}_{\mathrm{SUM}}(y,A\cup\{t\})
-\mathcal{P}_{\mathrm{SUM}}(y,A)
=
P(y,t),
\end{align}
independent of the current set. Hence
$\mathcal{P}_{\mathrm{SUM}}(y,\cdot)$ is monotone and modular.

For the maximum surrogate, $A\subseteq B$ implies
\begin{align}
\mathcal{P}_{\mathrm{MAX}}(y,A)
\leq
\mathcal{P}_{\mathrm{MAX}}(y,B),
\end{align}
so it is monotone. Its marginal gain from adding $t$ to a set $C$ is
\begin{align}
\mathcal{P}_{\mathrm{MAX}}(y,C\cup\{t\})
-\mathcal{P}_{\mathrm{MAX}}(y,C)
=
\max\!\left\{0,\,
P(y,t)-\mathcal{P}_{\mathrm{MAX}}(y,C)\right\}.
\end{align}
Since the current maximum can only increase as $C$ grows, this
marginal gain can only decrease. Therefore
$\mathcal{P}_{\mathrm{MAX}}(y,\cdot)$ is submodular.
\end{proof}

In particular, the marginal gain of $\mathcal{P}_{\mathrm{SUM}}$ is
independent of the existing transmitter set, whereas the marginal gain
of $\mathcal{P}_{\mathrm{MAX}}$ decreases as the strongest existing
signal increases.

\begin{thm}\label{thm:concavity}
If $\bar W$ is concave and non-decreasing, then
$\mathcal{S}(\cdot;\mathcal{P})$ is monotone and submodular for each
$\mathcal{P}\in
\{\mathcal{P}_{\mathrm{MAX}},\mathcal{P}_{\mathrm{SUM}}\}$.
\end{thm}

\begin{proof}[Proof of Theorem~\ref{thm:concavity}]
Monotonicity follows from monotonicity of $\mathcal{P}$ and
non-decreasingness of $\bar W$. To prove submodularity, fix
$A\subseteq B\subseteq\mathcal{X}$,
$t\in\mathcal{X}\setminus B$, and $y\in\Omega$. Set
\begin{align}
a&=\mathcal{P}(y,A),
&
b&=\mathcal{P}(y,B),\\
\delta_A
&=\mathcal{P}(y,A\cup\{t\})-\mathcal{P}(y,A),
&
\delta_B
&=\mathcal{P}(y,B\cup\{t\})-\mathcal{P}(y,B).
\end{align}
By Lemma~\ref{lem:P_submodular},
\begin{align}
a\leq b,
\qquad
\delta_A\geq\delta_B\geq0.
\end{align}

For $s,\delta\geq0$, define the increment
\begin{align}
D(s,\delta)=\bar W(s+\delta)-\bar W(s).
\end{align}
Because $\bar W$ is concave, $D(s,\delta)$ is non-increasing in $s$
for fixed $\delta$; because $\bar W$ is non-decreasing,
$D(s,\delta)$ is non-decreasing in $\delta$ for fixed $s$. Hence
\begin{align}
D(a,\delta_A)
\geq
D(a,\delta_B)
\geq
D(b,\delta_B).
\end{align}
Equivalently,
\begin{align}
&\bar W(\mathcal{P}(y,A\cup\{t\}))
-\bar W(\mathcal{P}(y,A))
\\
&\qquad\geq
\bar W(\mathcal{P}(y,B\cup\{t\}))
-\bar W(\mathcal{P}(y,B)).
\end{align}
Integrating this pointwise inequality against the non-negative density
$f(y)$ and using Theorem~\ref{thm:S_formula} gives
\begin{align}
\mathcal{S}(A\cup\{t\};\mathcal{P})
-\mathcal{S}(A;\mathcal{P})
\geq
\mathcal{S}(B\cup\{t\};\mathcal{P})
-\mathcal{S}(B;\mathcal{P}),
\end{align}
which proves submodularity.
\end{proof}

Theorem~\ref{thm:concavity} is the structural result used by the
greedy algorithm below. Concavity of $\bar W$ gives diminishing
utility to further improvements where the current aggregated network
quality is already high; thus equal increases in aggregated quality
are valued more strongly in comparatively poorly served regions.
This statement concerns the surrogate objective and does not, by
itself, imply geometric separation from existing transmitters or a
reduction in physical interference.

\begin{rem}[Usefulness of Theorem~\ref{thm:concavity}]
\label{rem:usefulness}
The primary role of Theorem~\ref{thm:concavity} is to provide the
diminishing-returns property used by
Proposition~\ref{prop:dim_returns} and
Theorem~\ref{thm:greedy_bound}. There is also a complementary
continuous relaxation for
$\mathcal{P}=\mathcal{P}_{\mathrm{SUM}}$. If the discrete transmitter
set is replaced by a non-negative measure $\nu$ on $\Omega$, then
$\mathcal{P}_{\mathrm{SUM}}(X,\nu)$ is linear in $\nu$ and
\begin{align}
\mathcal{S}(\nu;\mathcal{P}_{\mathrm{SUM}})
=
\mathbb{E}_X\!\left[
\bar W(\mathcal{P}_{\mathrm{SUM}}(X,\nu))
\right]
\end{align}
is concave in $\nu$. Consequently, maximizing this functional subject
to a convex mass constraint is a convex optimization problem; likewise,
minimizing total mass subject to a superlevel constraint on
$\mathcal{S}$ is convex. Concavity alone does not imply uniqueness of
the optimizer. Because the relaxation enlarges the feasible class, it
provides an upper bound on the discrete optimum for the maximization
problem and a lower bound on the minimum required mass for the
minimum-cardinality analogue. This relaxation is not used in the
numerical experiments. Extending a comparable continuous formulation
to $\mathcal{P}_{\mathrm{MAX}}$ is left for future work.
\end{rem}

The remainder of the paper focuses on the discrete candidate-set formulation used in practice. We next present IA-SPA, which selects transmitter locations sequentially and admits an explicit approximation guarantee for the resulting submodular objective.

\section{Proposed Framework and Algorithm}\label{sec:alg}

For a fixed aggregation measure
$\mathcal{P}\in\{\mathcal{P}_{\mathrm{MAX}},\mathcal{P}_{\mathrm{SUM}}\}$,
we denote the resulting algorithm by $\IASPA{\mathcal{P}}$. For a
current deployed set $T$, the marginal gain of adding a feasible site
$x$~is
\begin{align}
    \mathcal{G}_{\mathcal{P}}(x\mid T)
    =
    \mathcal{S}(T\cup\{x\};\mathcal{P})
    -
    \mathcal{S}(T;\mathcal{P}).
\end{align}

Algorithm~\ref{alg:eps-greedy} is written to cover both clean-slate and
incremental deployment. The set $T_{\mathrm{fixed}}$ contains sites
that are already deployed and is empty in the clean-slate case; the
set $A$ contains the newly selected sites.

\begin{algorithm}[H]
\caption{IA-SPA with aggregation measure $\mathcal{P}$}
\label{alg:eps-greedy}
\begin{algorithmic}
\State \textbf{Input:} feasible set $\mathcal{X}$, fixed set
$T_{\mathrm{fixed}}$, aggregation measure $\mathcal{P}$,
parameter $\epsilon\in[0,1)$, termination rule
\State $A\gets\emptyset$
\While{not $\mathrm{Terminated}(A)$ and
$\mathcal{X}\setminus(T_{\mathrm{fixed}}\cup A)\neq\emptyset$}
    \State $T\gets T_{\mathrm{fixed}}\cup A$
    \State Compute $\mathcal{G}_{\mathcal{P}}(x\mid T)$ for every
    $x\in\mathcal{X}\setminus T$
    \State $\mathcal{X}_{\epsilon}(T)\gets
    \left\{x\in\mathcal{X}\setminus T:
    \mathcal{G}_{\mathcal{P}}(x\mid T)
    \geq
    (1-\epsilon)
    \max\limits_{z\in\mathcal{X}\setminus T}
    \mathcal{G}_{\mathcal{P}}(z\mid T)
    \right\}$
    \State Choose $x^*\in\mathcal{X}_{\epsilon}(T)$ uniformly at random
    \State $A\gets A\cup\{x^*\}$
\EndWhile
\State \Return $T_{\mathrm{fixed}}\cup A$
\end{algorithmic}
\end{algorithm}

For the threshold problem~\eqref{eq:min_t_st_cov}, the clean-slate
termination rule is
$\mathcal{S}(A;\mathcal{P})\geq\beta$; with an existing deployment it
becomes
$\mathcal{S}(T_{\mathrm{fixed}}\cup A;\mathcal{P})\geq\beta$.
For a budget of $\beta$ newly added sites, the termination rule is
$|A|\geq\beta$. Because each iteration adds a previously unselected
candidate, the budget rule terminates after exactly $\beta$ additions,
while the threshold rule terminates after at most $|\mathcal{X}|$
additions. If the loop exhausts all candidates before the threshold is
reached, the target is infeasible for the prescribed candidate set.

For the theoretical guarantee below we take
$T_{\mathrm{fixed}}=\emptyset$. A fixed initial deployment can be
handled by the shifted objective
\begin{align}
F(A)
=
\mathcal{S}(T_{\mathrm{fixed}}\cup A;\mathcal{P})
-
\mathcal{S}(T_{\mathrm{fixed}};\mathcal{P}),
\end{align}
which remains monotone and submodular in $A$. Hence the same
diminishing-returns analysis applies to incremental deployment, with
the guarantee interpreted relative to the gain above the fixed
baseline.



\subsection{Theoretical Guarantee on Network Quality}

\begin{prop}\label{prop:dim_returns}
Fix
$\mathcal{P}\in
\{\mathcal{P}_{\mathrm{MAX}},\mathcal{P}_{\mathrm{SUM}}\}$.
Under the assumptions of Theorem~\ref{thm:concavity}, for
$T_1\subseteq T_2\subseteq\mathcal{X}$ and
$x\in\mathcal{X}\setminus T_2$,
\begin{align}
    \mathcal{G}_{\mathcal{P}}(x\mid T_1)
    \geq
    \mathcal{G}_{\mathcal{P}}(x\mid T_2).
\end{align}
\end{prop}

\begin{proof}[Proof of Proposition~\ref{prop:dim_returns}]
This is exactly the submodularity inequality for
$\mathcal{S}(\cdot;\mathcal{P})$ established in
Theorem~\ref{thm:concavity}.
\end{proof}

\begin{thm}\label{thm:greedy_bound}
Fix
$\mathcal{P}\in
\{\mathcal{P}_{\mathrm{MAX}},\mathcal{P}_{\mathrm{SUM}}\}$ and
$\epsilon\in[0,1)$. Let
\begin{align}
T_k^*
\in
\argmax_{\substack{T\subseteq\mathcal{X}\\ |T|=k}}
\mathcal{S}(T;\mathcal{P}),
\end{align}
and let $T_n$ be the $n$-site clean-slate deployment generated by
$\IASPA{\mathcal{P}}$, with $0\leq n\leq|\mathcal{X}|$. Then
\begin{align}
\left(
1-e^{-\frac{n(1-\epsilon)}{k}}
\right)
\mathcal{S}(T_k^*;\mathcal{P})
\leq
\mathcal{S}(T_n;\mathcal{P}).
\end{align}
\end{thm}

\begin{proof}[Proof of Theorem~\ref{thm:greedy_bound}]
Let
$T_k^*=\{y_1^*,\ldots,y_k^*\}$. By monotonicity and then
submodularity,
\begin{align}
\mathcal{S}(T_k^*;\mathcal{P})
&\leq
\mathcal{S}(T_n\cup T_k^*;\mathcal{P})\\
&\leq
\mathcal{S}(T_n;\mathcal{P})
+
\sum_{j=1}^k
\mathcal{G}_{\mathcal{P}}(y_j^*\mid T_n).
\end{align}
If $y_j^*\in T_n$, the corresponding marginal gain is zero; otherwise
it is bounded by the largest available marginal gain. Hence
\begin{align}
\mathcal{S}(T_k^*;\mathcal{P})
-
\mathcal{S}(T_n;\mathcal{P})
\leq
k
\max_{x\in\mathcal{X}\setminus T_n}
\mathcal{G}_{\mathcal{P}}(x\mid T_n).
\end{align}
By the selection rule of $\IASPA{\mathcal{P}}$,
\begin{align}
\mathcal{S}(T_{n+1};\mathcal{P})
-
\mathcal{S}(T_n;\mathcal{P})
=
\mathcal{G}_{\mathcal{P}}(x_{n+1}\mid T_n)
\geq
(1-\epsilon)
\max_{x\in\mathcal{X}\setminus T_n}
\mathcal{G}_{\mathcal{P}}(x\mid T_n).
\end{align}
Therefore, with
\begin{align}
\Delta_n
=
\mathcal{S}(T_k^*;\mathcal{P})
-
\mathcal{S}(T_n;\mathcal{P}),
\end{align}
we have
\begin{align}
\Delta_{n+1}
\leq
\left(
1-\frac{1-\epsilon}{k}
\right)
\Delta_n
\leq
e^{-\frac{1-\epsilon}{k}}\Delta_n.
\end{align}
Since
$\mathcal{S}(\emptyset;\mathcal{P})=0$,
iteration yields
\begin{align}
\Delta_n
\leq
e^{-\frac{n(1-\epsilon)}{k}}
\mathcal{S}(T_k^*;\mathcal{P}),
\end{align}
which is equivalent to the stated bound.
\end{proof}

For $n=k$ and $\epsilon=0$, Theorem~\ref{thm:greedy_bound} gives the
classical $1-1/e\approx63.2\%$ approximation guarantee for the
surrogate objective $\mathcal{S}(\cdot;\mathcal{P})$. This is a worst-case surrogate bound, not a quantitative prediction for the SINR metrics below.

\subsection{Computational Complexity}\label{sec:complexity}

Let $\mathcal{X}$ be the set of candidate transmitter locations and
$\Omega_{\delta}$ the discretized receiver grid. We assume that the
propagation fields $P(y,x)$ have been pre-computed for all
$x \in \mathcal{X}$ and $y \in \Omega_{\delta}$. By
Theorem~\ref{thm:S_formula}, the discretized functional
$\mathcal{S}(T;\mathcal{P})$ is evaluated by applying $\bar{W}$ pointwise to the
aggregated field $\mathcal{P}(y,T)$ and averaging over
$\Omega_{\delta}$; no discretization of the auxiliary variable
$\kappa$ is required.

For both aggregation measures considered here, the current aggregated
field can be maintained throughout the greedy iteration. If
$Q_T(y) = \mathcal{P}(y,T)$, then adding a candidate $x$ gives
\begin{align}
    Q_{T \cup \{x\}}(y)
    &= Q_T(y) + P(y,x),
    && \text{for } \mathcal{P}_{\mathrm{SUM}}, \\
    Q_{T \cup \{x\}}(y)
    &= \max\{Q_T(y),P(y,x)\},
    && \text{for } \mathcal{P}_{\mathrm{MAX}}.
\end{align}
Hence the marginal gain of one candidate can be evaluated in
$O(|\Omega_{\delta}|)$ operations. Evaluating all
$|\mathcal{X}|$ candidates therefore costs
$O(|\mathcal{X}|\,|\Omega_{\delta}|)$ per greedy iteration, and
selecting $n$ transmitters can be implemented in
\begin{align}
    O\!\left(n |\mathcal{X}|\,|\Omega_{\delta}|\right)
\end{align}
operations once the propagation maps have been computed. Storing all
pre-computed maps requires
$O(|\mathcal{X}|\,|\Omega_{\delta}|)$ memory.

\subsection{Empirical Runtime and Scalability}\label{sec:runtime}

The computational pipeline consists of two stages: ray-traced
radio-map generation and greedy site selection. For each candidate
$x \in \mathcal{X}$, the radio map $P(\cdot,x)$ is generated once,
independently of the other candidates, and is subsequently reused
throughout the optimization. 
Consequently, radio-map generation is embarrassingly parallel over the
candidate sites $x \in \mathcal{X}$. For a fixed receiver grid and
ray-tracing configuration, the total work of this stage grows linearly
with $|\mathcal{X}|$, since one independent propagation simulation is
performed for each candidate. Its wall-clock time can therefore be
reduced substantially by distributing these independent simulations
across available parallel hardware. In particular, ray tracing need
not be repeated for different multi-transmitter configurations.

\begin{table}[t]
    \centering
    \renewcommand{\arraystretch}{0.6}
    \caption{Wall-clock runtime for radio-map generation (GPU) and
    greedy site selection (CPU).}
    \label{tab:runtime}
    \footnotesize
    \begin{tabular}{l r r r r}
        \toprule
        \textbf{Scene} & $|\mathcal{X}|$ &
        \textbf{Map / cand.} & \textbf{All maps} &
        \textbf{Greedy / iter.} \\
        \midrule
        San Francisco & 11207 & 0.07 s & 0.22 h & 69.2 s \\
        Florence      & 10836 & 0.11 s & 0.33 h & 45.4 s \\
        \bottomrule
    \end{tabular}
\end{table}

Table~\ref{tab:runtime} reports wall-clock timings on a single NVIDIA
Tesla V100-PCIE-16GB for radio-map generation and a single core of an
Intel Xeon Gold 6248R @ 3.00\,GHz for the CPU-only greedy optimizer.
Generating the approximately $11{,}000$ candidate radio maps requires
$0.22$ GPU-hours for San Francisco and $0.33$ GPU-hours for Florence.
Once these maps have been generated, each greedy iteration requires
$45$--$70$ seconds, so placing 33 transmitters from scratch takes less
than 40 minutes of CPU time. Including radio-map generation, the
complete pipeline finishes in under one hour for either city on the
hardware used here. These measurements demonstrate that the method is
computationally feasible at the city-scale candidate-set sizes
considered in our experiments.

\section{Numerical Results}\label{sec:num_results}
We evaluate IA-SPA using site-specific ray-tracing simulations in
San Francisco and Florence, and compare the resulting deployments with
reference operator configurations and several constrained and
incremental deployment scenarios.

\subsection{Simulation Setup}\label{sec:greedy_num_res_setup}

We conduct numerical experiments using the San Francisco and Florence
3D urban scenes provided in the \texttt{Sionna} ray-tracing framework
(Fig.~\ref{fig:sf_map}). The scenes contain site-specific building
geometry and material properties used to generate the propagation
fields $P(y,x)$. Unless stated otherwise, the carrier frequency is
$1.8$ GHz, the transmit power is $40$ dBm ($10$ W), and the bandwidth
is $B=10$ MHz.

The San Francisco scene covers approximately $1.1\times1.1$ km, while
the Florence scene covers approximately $1.0\times1.1$ km. Candidate
transmitter sites $\mathcal{X}$ are generated on a uniform
$10\,\mathrm{m}\times10\,\mathrm{m}$ horizontal grid at $20$ m above
the local terrain and filtered using the scene geometry to remove
locations inside buildings. This yields
$|\mathcal{X}|=11{,}207$ feasible sites in San Francisco and
$|\mathcal{X}|=10{,}836$ in Florence. Receiver locations are sampled
over the same horizontal domains at heights $1.5$, $5$, $10$, and
$20$ m. Performance statistics are computed separately at each
receiver height and, unless stated otherwise, the four resulting
values are averaged.

Reference operator deployments are constructed from reported
transmitter coordinates obtained from the City and County of San
Francisco open-data portal~\cite{sfdata2025} and the OpenCellID
database~\cite{OpenCellID2026}. The comparisons below use the AT\&T and
T-Mobile deployments in San Francisco and the Iliad deployment in
Florence.

\begin{figure}[h!]
    \centering
    \includegraphics[width=0.49\linewidth]{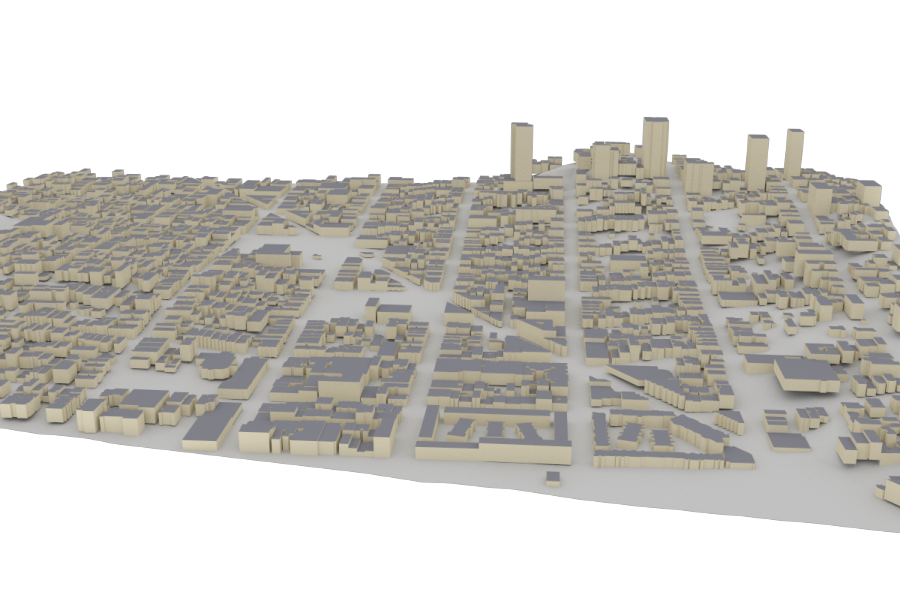}
    \includegraphics[width=0.49\linewidth]{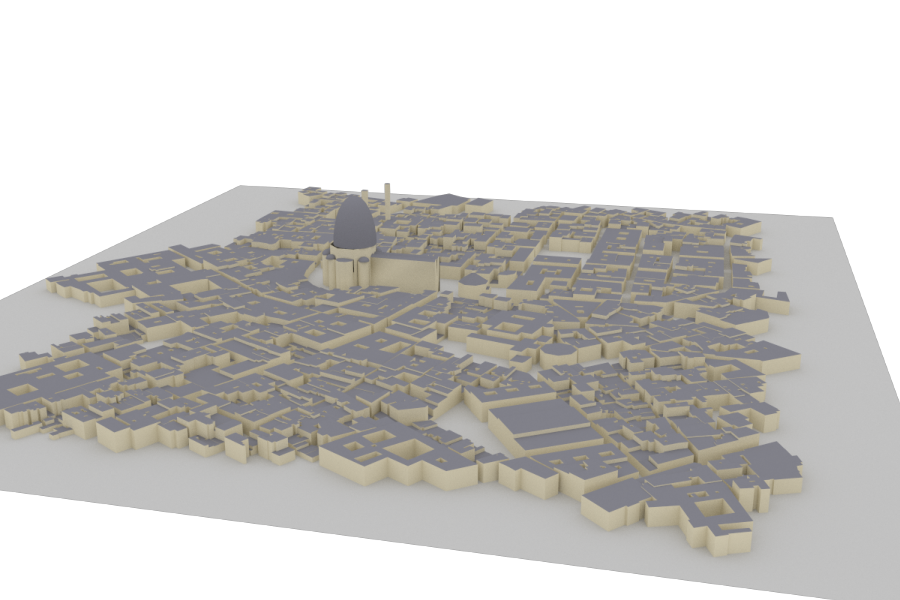}
    \caption{3D maps of San Francisco (left) and Florence (right)
    included in Sionna.}
    \label{fig:sf_map}
\end{figure}

\subsection{Performance Metrics}\label{sec:performance_metrics}

Although the placement is computed using either
$\IASPA{\mathcal{P}_{\mathrm{MAX}}}$ or
$\IASPA{\mathcal{P}_{\mathrm{SUM}}}$, every
resulting deployment is evaluated using the
physical SINR defined in Section~\ref{sec:systemmodel}. In terms of the
two aggregated power measures,
\begin{equation}
    \mathrm{SINR}(y,T)
    =
    \frac{\mathcal{P}_{\mathrm{MAX}}(y,T)}
    {\mathcal{P}_{\mathrm{SUM}}(y,T)
     -\mathcal{P}_{\mathrm{MAX}}(y,T)+\sigma^2}.
\end{equation}
The corresponding achievable data rate is
\begin{equation}
    C(y)
    =
    B\log_2\left(
    1+\frac{\mathrm{SINR}(y,T)}{\Gamma}
    \right),
\end{equation}
where $B=10$ MHz and $\Gamma=2$ (3 dB) is the implementation gap.

At each receiver height, we report the mean, standard deviation, and
maximum of the achieved data rate. We also report the \emph{edge rate},
defined as the 5th percentile of $C(y)$ over the covered receiver
locations at that height,
\begin{equation}
    \{y\in\Omega_\delta : C(y)>0\}.
\end{equation}
Locations with zero achievable rate, including locations inside
buildings or otherwise unreachable from the deployment, are therefore
not counted as zero-rate samples in the edge-rate statistic. Each of
these four rate statistics is computed separately at heights
$1.5$, $5$, $10$, and $20$ m, and the reported value is the arithmetic
mean of the four height-specific values.

We also report the interference power
\begin{equation}
    I(y,T)
    =
    \mathcal{P}_{\mathrm{SUM}}(y,T)
    -\mathcal{P}_{\mathrm{MAX}}(y,T),
\end{equation}
which is the total received power from all transmitters other than the
strongest serving transmitter under the association model of
Section~\ref{sec:systemmodel}. The mean, standard deviation, and
maximum interference are likewise computed separately at each receiver
height and then averaged over the four heights.

\subsection{San Francisco with $\IASPA{\mathcal{P}_\text{MAX}}$}

We first evaluate
$\IASPA{\mathcal{P}_\text{MAX}}$ and compare the
resulting deployments with the
AT\&T and T-Mobile reference configurations in San Francisco.
Figures~\ref{fig:greedy_res_max} and~\ref{fig:greed_interf_max}
illustrate the achieved data-rate and interference fields,
respectively, for the AT\&T comparison.

\begin{figure}[H]
    \centering
    \includegraphics[width=\linewidth]{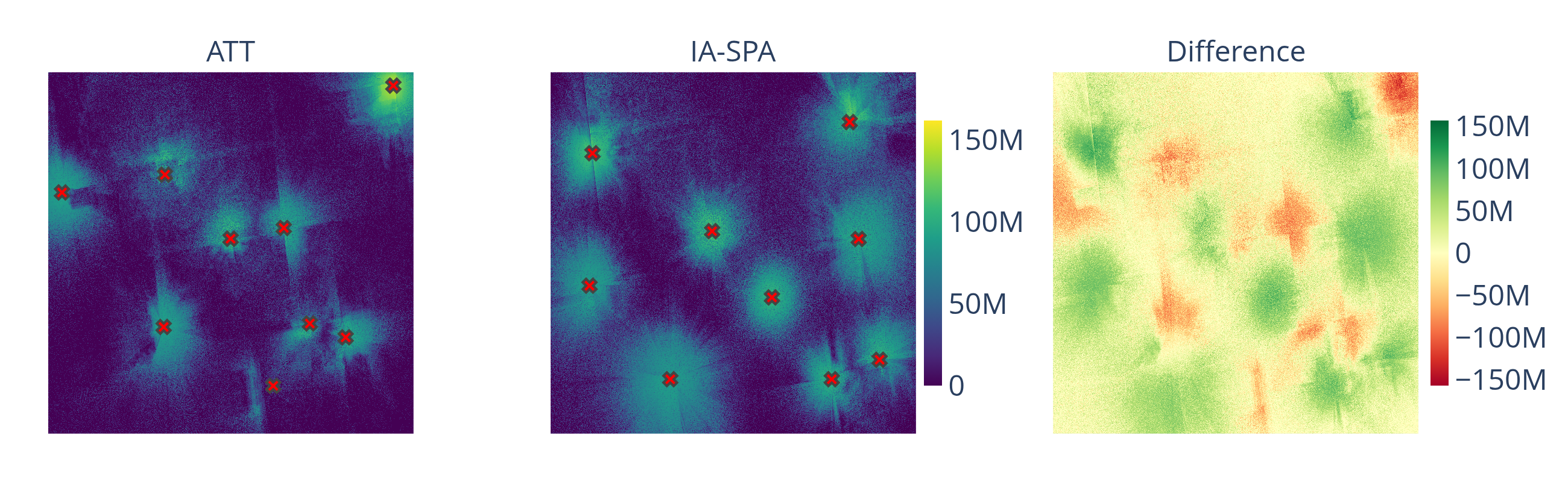}
    \caption{Achieved data rate across San Francisco for the AT\&T
    reference deployment and the deployment produced by
    $\IASPA{\mathcal{P}_\text{MAX}}$.}
    \label{fig:greedy_res_max}
\end{figure}

\begin{figure}[H]
    \centering
    \includegraphics[width=\linewidth]{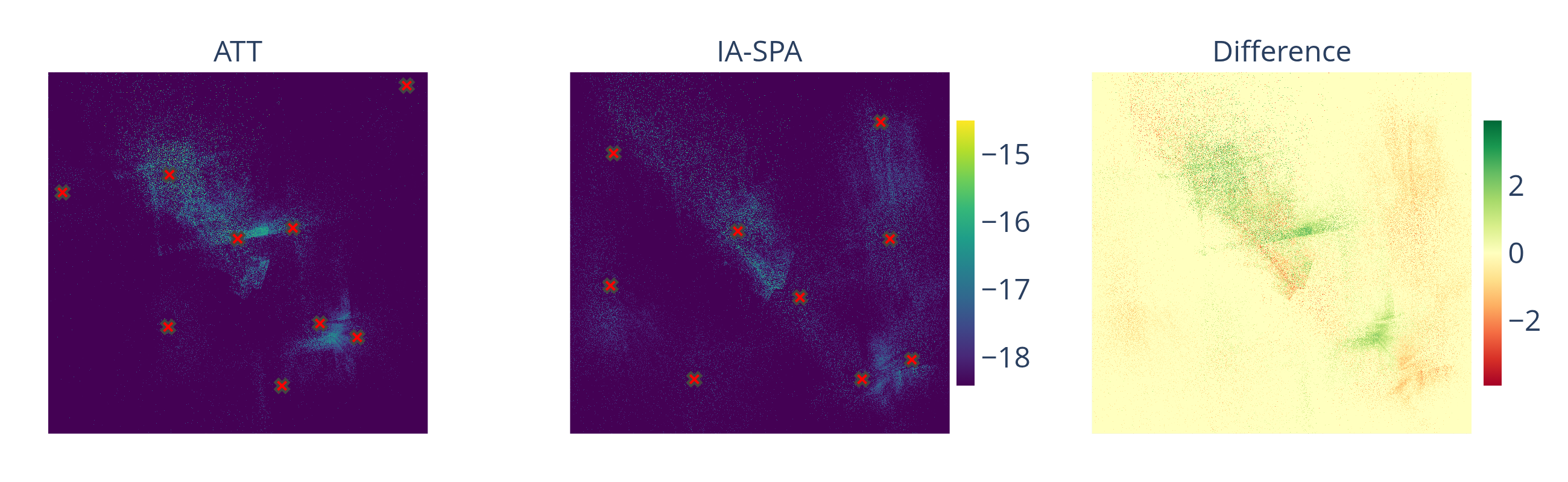}
    \caption{Interference across San Francisco for the AT\&T reference
    deployment and the deployment produced by
    $\IASPA{\mathcal{P}_\text{MAX}}$.}
    \label{fig:greed_interf_max}
\end{figure}

Table~\ref{tab:comparison_max} summarizes the quantitative results.
For the AT\&T reference, $\IASPA{\mathcal{P}_\text{MAX}}$ increases the mean data rate by
$70.4\%$ and the edge rate by $88.0\%$. For the T-Mobile reference,
the corresponding improvements are $215.5\%$ and $758.4\%$.
At the same time, mean interference decreases by $2.0\%$ in the
AT\&T comparison and by $93.7\%$ in the T-Mobile comparison. Thus,
the throughput gains are not obtained simply by increasing useful
received power at the expense of additional interference.

\begin{table*}[t]
    \centering
    \renewcommand{\arraystretch}{1.2}
    \caption{Performance comparison in San Francisco for
    $\IASPA{\mathcal{P}_\text{MAX}}$.}
    \label{tab:comparison_max}
    \small
    \begin{tabular}{l r r r r r r}
        \toprule
        & \multicolumn{3}{c}{\textbf{AT\&T Scenario}}
        & \multicolumn{3}{c}{\textbf{T-Mobile Scenario}} \\
        \cmidrule(lr){2-4} \cmidrule(lr){5-7}
        \textbf{Data Rate [Mbps]}
        & \textbf{Reference} & \textbf{IA-SPA} & \textbf{Change}
        & \textbf{Reference} & \textbf{IA-SPA} & \textbf{Change} \\
        \midrule
        Mean Rate
        & 21.18 & 36.10 & +70.4\%
        & 10.06 & 31.72 & +215.49\% \\
        Std.\ Deviation
        & 26.21 & 28.84 & +10.01\%
        & 20.24 & 29.75 & +47.01\% \\
        Max Rate
        & 146.35 & 116.50 & -20.4\%
        & 154.24 & 125.72 & -18.49\% \\
        Edge Rate (5\%-tile)
        & 5.51 & 10.36 & +88.0\%
        & 1.11 & 9.55 & +758.39\% \\
        \midrule
        \textbf{Interference [nW]}
        & \textbf{Reference} & \textbf{IA-SPA} & \textbf{Change}
        & \textbf{Reference} & \textbf{IA-SPA} & \textbf{Change} \\
        \midrule
        Mean Interference
        & 2.24 & 2.19 & -2.04\%
        & 17.54 & 1.11 & -93.70\% \\
        Std.\ Deviation
        & 11.53 & 8.58 & -25.57\%
        & 216.76 & 5.22 & -97.59\% \\
        Max Interference
        & 495.86 & 473.93 & -4.42\%
        & 20491.9 & 214.04 & -98.96\% \\
        \bottomrule
    \end{tabular}
\end{table*}

The maximum data rate is lower for $\IASPA{\mathcal{P}_\text{MAX}}$ in both comparisons, despite
the substantial improvements in mean and edge rates. This is not
inconsistent with the network-wide gains: the maximum rate is
determined by a single receiver location and can therefore be dominated
by an isolated high-SINR region. The mean and edge rates provide more
representative measures of performance across the receiver domain.

For comparison, we also optimized the San Francisco deployments using \(\mathcal P_{\rm SUM}\). For the AT\&T reference, the resulting mean- and edge-rate improvements are \(43.8\%\) and \(37.1\%\), respectively, compared with \(70.4\%\) and \(88.0\%\) using \(\mathcal P_{\rm MAX}\). The same ordering is observed for the T-Mobile comparison. These results indicate that, among the two surrogates considered here, \(\mathcal P_{\rm MAX}\) is better aligned empirically with the final SINR-based performance metric; we therefore use \(\mathcal P_{\rm MAX}\) in the remaining experiments.

\subsection{Florence with Exclusion Zones}\label{sec:exclusion}

We next examine the ability of
$\IASPA{\mathcal{P}_{\mathrm{MAX}}}$ to
accommodate spatial constraints on transmitter placement in the
Florence scene. An exclusion
zone is imposed by removing from the feasible candidate set
$\mathcal{X}$ all candidate sites whose locations lie within the
restricted region. The propagation fields $P(y,x)$ themselves are
left unchanged. Such constraints can represent, for example,
restrictions arising from real-estate availability, permitting, or
heritage-preservation requirements in historic urban areas.

We consider two levels of restriction. The first experiment imposes a
single elliptical exclusion zone around the Basilica di San Lorenzo
and its surrounding piazza (representing the type of siting restriction
that can arise around protected historic sites in Italy), in which new
transmitter placement is prohibited. The second imposes three
non-contiguous exclusion zones in different parts of Florence,
providing a more restrictive test of the placement algorithm. The three
zones are shown in Fig.~\ref{fig:exp1_zones}.

\subsubsection{Single Exclusion Zone}\label{sec:singlezone}

Table~\ref{tab:fl_metrics} compares the deployment produced by $\IASPA{\mathcal{P}_{\mathrm{MAX}}}$ with the
Iliad reference deployment under the single exclusion-zone constraint.
$\IASPA{\mathcal{P}_{\mathrm{MAX}}}$ increases the mean data rate by approximately $30\%$ and the
edge rate by approximately $56\%$, while reducing mean interference by
approximately $30\%$. Thus, substantial improvements over the reference
deployment remain possible even after candidate sites in the protected
region are removed from the feasible set.

\begin{table}[ht!]
\centering
\renewcommand{\arraystretch}{1.2}
\caption{Performance comparison in Florence with a single exclusion
zone using $\IASPA{\mathcal{P}_{\mathrm{MAX}}}$.}
\label{tab:fl_metrics}
\begin{tabular}{@{}lrrr@{}}
\toprule
\textbf{Metric}
& \textbf{Iliad}
& \textbf{IA-SPA}
& \textbf{Change} \\
\midrule
\textit{Data Rate [Mbps]} & & & \\
Mean Data Rate
& 39.21 & 50.96 & +29.98\% \\
Standard Deviation
& 29.38 & 29.28 & -0.33\% \\
Maximum Data Rate
& 204.79 & 207.76 & +1.45\% \\
Edge Rate (5\%-tile)
& 7.41 & 11.53 & +55.62\% \\
\midrule
\textit{Interference [nW]} & & & \\
Mean Interference
& 85.44 & 59.40 & -30.48\% \\
Standard Deviation
& 182.05 & 69.89 & -61.61\% \\
Maximum Interference
& 4998.62 & 2490.50 & -50.18\% \\
\bottomrule
\end{tabular}
\end{table}

\subsubsection{Multiple Non-Contiguous Exclusion Zones}
\label{sec:multizone}

We next impose three non-contiguous elliptical exclusion zones: the
original Basilica di San Lorenzo zone, a second zone to the northwest,
and a third zone to the south, as shown in
Fig.~\ref{fig:exp1_zones}. Candidate sites lying inside any of the
three regions are removed from the feasible set before
$\IASPA{\mathcal{P}_{\mathrm{MAX}}}$ is applied.

For consistency, we also check whether any transmitter in the Iliad
reference deployment lies inside the imposed regions. None of the
reference sites is removed by the two additional zones, so the
reference deployment remains the same 33-transmitter configuration used in
the single-zone experiment. $\IASPA{\mathcal{P}_{\mathrm{MAX}}}$ is therefore also evaluated with
33 transmitters, but over the more restricted candidate set.

\begin{figure}[t]
    \centering
    \includegraphics[width=0.65\linewidth]{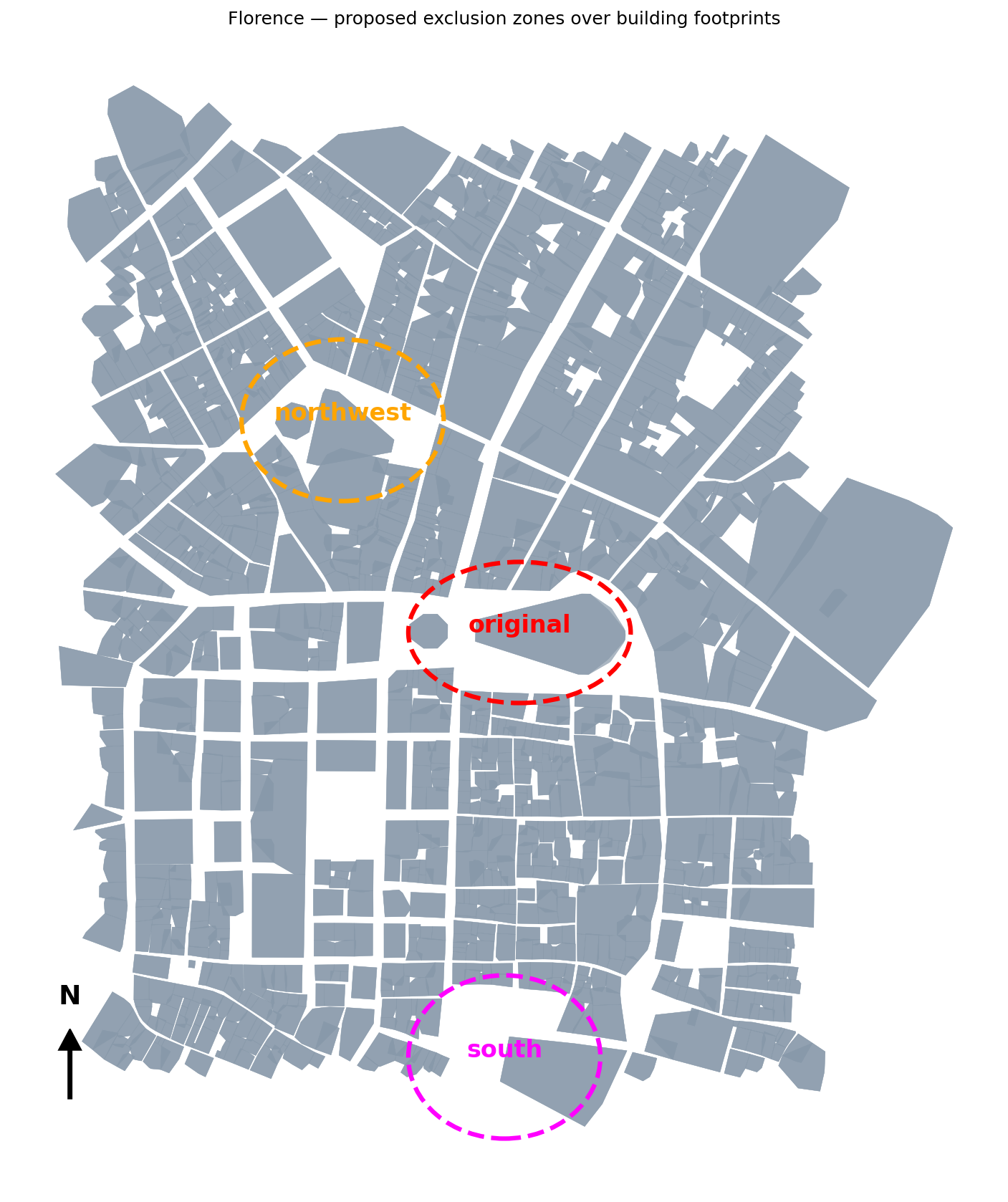}
    \caption{Florence building footprints and the three non-contiguous
    exclusion zones used in Section~\ref{sec:multizone}: the Basilica
    di San Lorenzo zone (red), together with additional zones to the
    northwest (orange) and south (magenta).}
    \label{fig:exp1_zones}
\end{figure}

\begin{table}[t]
\centering
\renewcommand{\arraystretch}{1.2}
\caption{Performance comparison in Florence with three non-contiguous
exclusion zones using $\IASPA{\mathcal{P}_{\mathrm{MAX}}}$
(33 transmitters in each deployment).}
\label{tab:fl_multi3}
\begin{tabular}{@{}lrrr@{}}
\toprule
\textbf{Metric}
& \textbf{Iliad}
& \textbf{IA-SPA}
& \textbf{Change} \\
\midrule
\textit{Data Rate [Mbps]} & & & \\
Mean Data Rate
& 39.15 & 49.77 & +27.13\% \\
Standard Deviation
& 29.23 & 29.00 & -0.80\% \\
Maximum Data Rate
& 204.78 & 207.76 & +1.45\% \\
Edge Rate (5\%-tile)
& 7.43 & 11.21 & +50.88\% \\
\midrule
\textit{Interference [nW]} & & & \\
Mean Interference
& 88.10 & 63.12 & -28.32\% \\
Standard Deviation
& 182.38 & 71.59 & -60.75\% \\
Maximum Interference
& 5064.36 & 1858.59 & -63.30\% \\
\bottomrule
\end{tabular}
\end{table}

\begin{figure}[t]
    \centering
    \includegraphics[width=\linewidth]{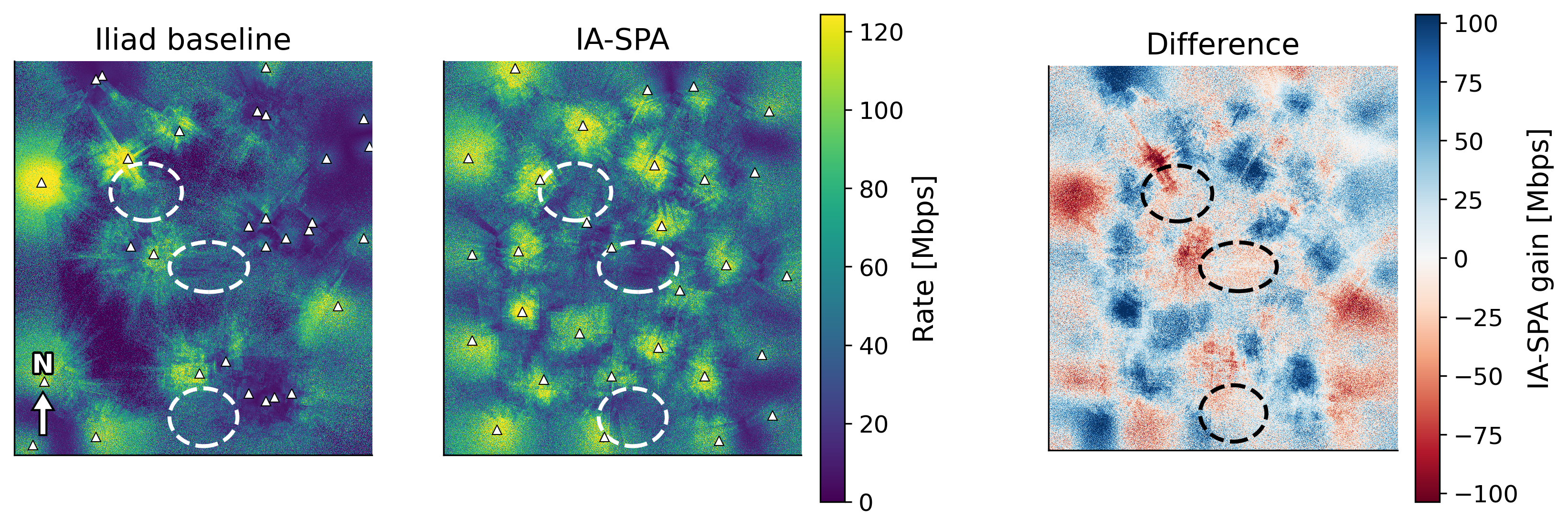}
    \caption{Data rate across Florence with three non-contiguous
    exclusion zones (dashed white ellipses) using
    $\IASPA{\mathcal{P}_{\mathrm{MAX}}}$:
    Iliad reference deployment (left), optimized deployment (center),
    and absolute gain (right). Triangles
    mark transmitter locations.}
    \label{fig:exp1_rate_multi3}
\end{figure}

\begin{figure}[t]
    \centering
    \includegraphics[width=\linewidth]{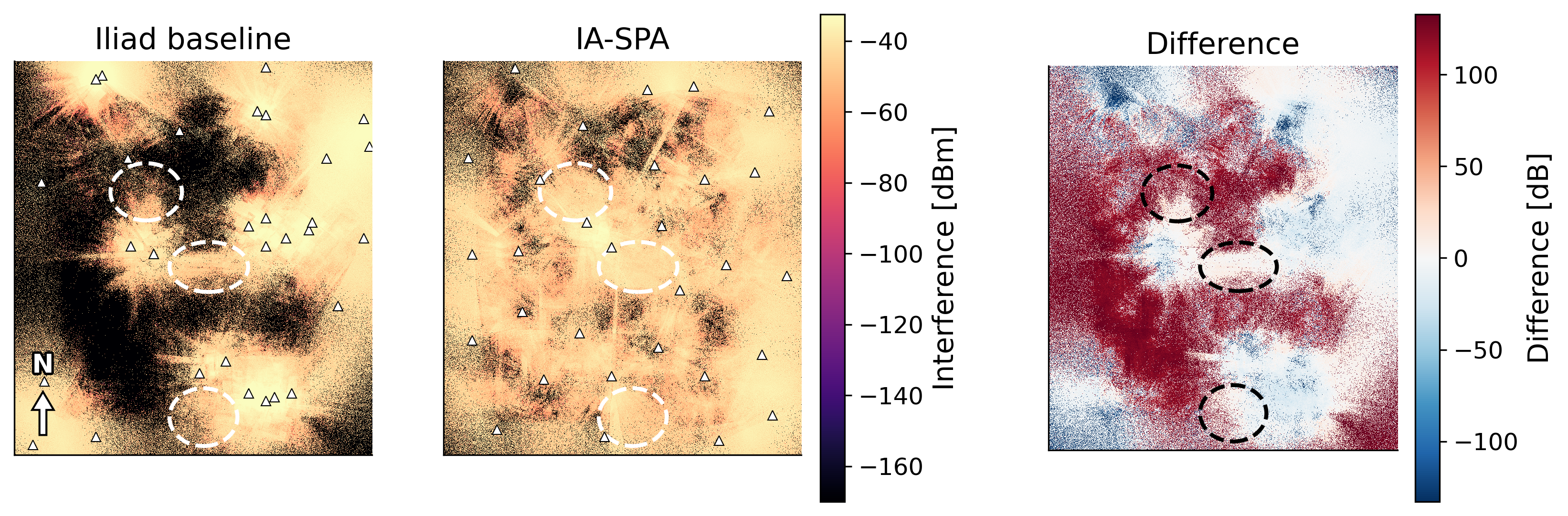}
    \caption{Interference across Florence with three non-contiguous
    exclusion zones (dashed white ellipses) using
    $\IASPA{\mathcal{P}_{\mathrm{MAX}}}$:
    Iliad reference deployment (left), optimized deployment (center),
    and absolute difference (right).}
    \label{fig:exp1_interf_multi3}
\end{figure}

As shown in Table~\ref{tab:fl_multi3}, $\IASPA{\mathcal{P}_{\mathrm{MAX}}}$ retains a substantial
advantage under the more restricted candidate set, increasing the mean
data rate by $27.1\%$ and the edge rate by $50.9\%$, while reducing
mean interference by $28.3\%$. Compared with the single-zone results
in Table~\ref{tab:fl_metrics}, the additional two exclusion zones
produce only a modest reduction in the $\IASPA{\mathcal{P}_{\mathrm{MAX}}}$ rate advantage: the
mean-rate improvement decreases from approximately $30\%$ to $27\%$,
and the edge-rate improvement from approximately $56\%$ to $51\%$.
For the particular exclusion regions considered here, the placement
advantage therefore remains substantial even when several separated
parts of the candidate space are unavailable. More generally, this
experiment illustrates that geographic siting restrictions arising
from permitting, real-estate, or heritage-preservation constraints can
be incorporated directly by modifying the feasible candidate set,
without changing the IA-SPA optimization framework.

\subsection{Robustness to Small-Scale Fading}\label{sec:fading}

The transmitter configurations considered above are optimized using the
deterministic site-specific propagation fields $P(y,x)$ generated by
Sionna RT. These fields account for the propagation effects resolved by
the ray tracer, but do not include additional random small-scale
fading. To test whether the performance advantage of
$\IASPA{\mathcal{P}_{\mathrm{MAX}}}$ persists
in the presence of such fading, we keep the transmitter locations
fixed and superimpose random fading on the received powers. We use the
San Francisco AT\&T scenario as the test case and compare the fixed
$\IASPA{\mathcal{P}_{\mathrm{MAX}}}$ and
AT\&T deployments.

We consider two models for the additional fading coefficient: Rayleigh
fading and Rician fading with $K=10$ dB. 
We also consider \(L\in\{1,2,4\}\) independent fading branches to model an effective diversity order. As specified below, the branch SINRs are averaged before evaluating the rate; thus, \(L\) should be interpreted as a generic diversity-order proxy rather than as a particular receive-antenna combining architecture.
For each fading model and each
value of $L$, we generate $1000$ independent realizations and recompute
the SINR, mean data rate, and edge rate for both deployments.

Each transmitter-receiver link is faded independently
over the $L$ branches. For branch $\ell \in \{1,\dots,L\}$, transmitter
$x \in T$, and receiver $y$, the deterministic field is scaled by an
independent, unit-power Rician gain,
\begin{equation}
P_{\mathrm{fad},\ell}(y,x) = |h_\ell(y,x)|^2\, P(y,x), \qquad
h_\ell(y,x) \sim \sqrt{\tfrac{K}{K+1}} + \mathcal{CN}\!\left(0,\tfrac{1}{K+1}\right),
\end{equation}
so that $\mathbb{E}\big[|h_\ell(y,x)|^2\big] = 1$ and
$\mathbb{E}\big[P_{\mathrm{fad},\ell}(y,x)\big] = P(y,x)$ for every branch
(Rayleigh fading is the $K=0$ case). This confirms the $L=1$ expression above.
Every transmitter in $T$, serving and interfering alike, fades independently
per branch, while the serving transmitter itself is fixed by the
\emph{unfaded} field, $t^\star(y) = \arg\max_{x \in T} P(y,x)$, i.e.\
association is on long-term average power rather than on the instantaneous
fade. For $L>1$ the branch powers are \emph{not} summed and renormalized;
instead each branch yields its own instantaneous SINR,
\begin{equation}
\mathrm{SINR}_\ell(y) = \frac{P_{\mathrm{fad},\ell}(y,t^\star)}
{\sum_{x \in T} P_{\mathrm{fad},\ell}(y,x) - P_{\mathrm{fad},\ell}(y,t^\star) + N_0},
\end{equation}
where $N_0$ is the receiver thermal noise power over bandwidth $B$ (in
absolute units here, rather than the noise-normalized $\sigma^2=1$ convention
of Sec.~II-A), and the $L$ branch SINRs are averaged before the rate is
evaluated,
\begin{equation}
C(y) = B \log_2\!\left(1 + \frac{1}{L}\sum_{\ell=1}^{L} \mathrm{SINR}_\ell(y) \,/\, \Gamma\right).
\end{equation}
This SINR averaging is the sense in which the \(L\) branches are combined; \(L\) should therefore be interpreted as a generic diversity-order proxy rather than as a specific receive-combining architecture or multicarrier rate model.

\begin{figure}[t]
    \centering
    \includegraphics[width=\linewidth]{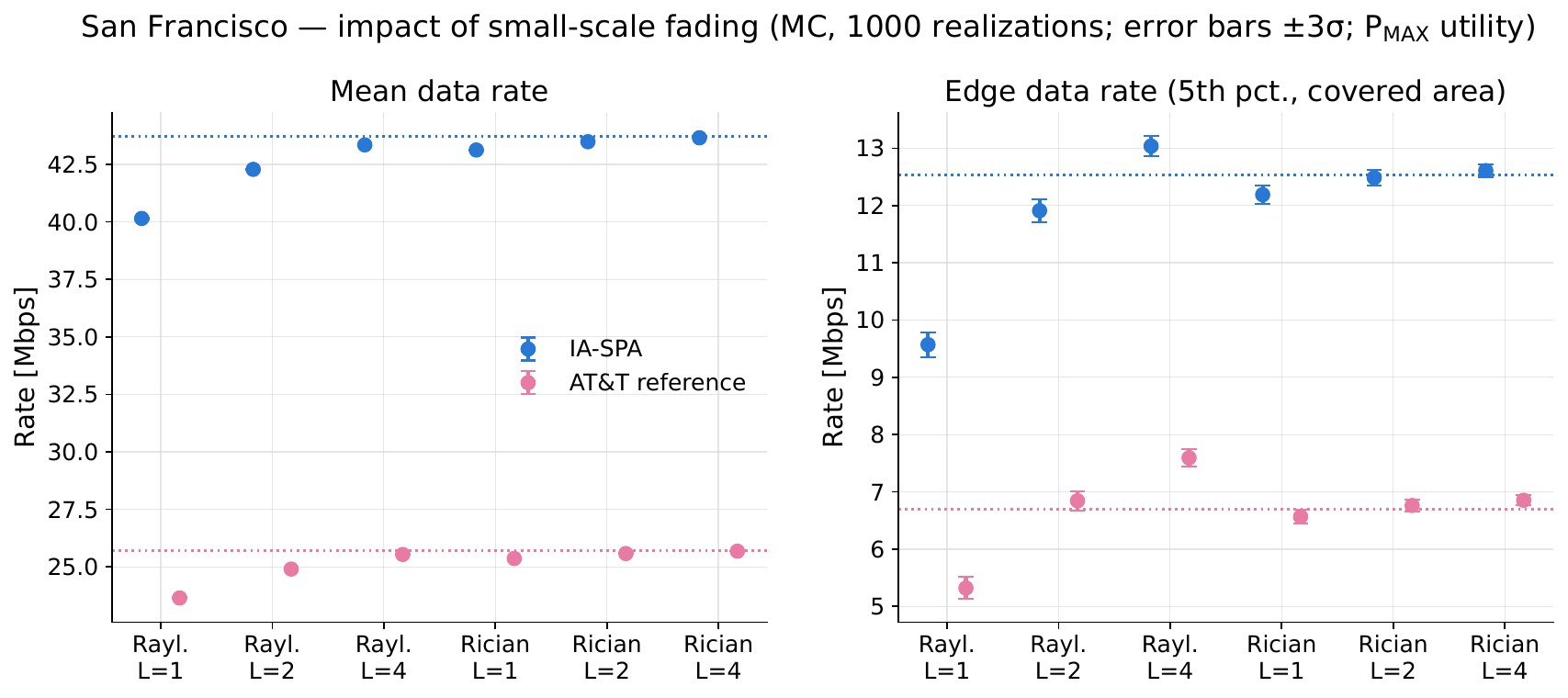}
    \caption{Impact of small-scale fading on the mean (left) and edge
    (right) data rates for the San Francisco AT\&T scenario using
    $\mathcal{P}_{\mathrm{MAX}}$. Dotted lines show the corresponding
    deterministic (no-fading) values from the same pipeline run.
    Markers show the Monte Carlo mean over $1000$ realizations, with
    error bars at $\pm3\sigma$.}
    \label{fig:exp2_fading}
\end{figure}

Figure~\ref{fig:exp2_fading} shows that fading has its largest effect
for Rayleigh fading with $L=1$. Relative to the deterministic baseline,
the $\IASPA{\mathcal{P}_{\mathrm{MAX}}}$ mean rate decreases by about $8\%$ (from $43.7$ to
$40.2$ Mbps), while the edge rate decreases by about $23\%$ (from
$12.5$ to $9.6$ Mbps). Increasing the diversity order reduces this
degradation, as does the presence of a dominant component in the Rician
model.

More importantly, the advantage of $\IASPA{\mathcal{P}_{\mathrm{MAX}}}$ over the AT\&T reference is
preserved in every fading configuration tested. Even for Rayleigh
fading with $L=1$, the $\IASPA{\mathcal{P}_{\mathrm{MAX}}}$ edge rate of $9.6$ Mbps remains about
$81\%$ higher than the AT\&T reference value of $5.3$ Mbps, compared
with an $88\%$ advantage in the deterministic case. Thus, within the
fading and diversity models considered here, the performance advantage
of the deployment obtained from the deterministic propagation model
persists after small-scale fading is introduced.

\subsection{Robustness to Model Uncertainty}\label{sec:sensitivity}

The propagation field $P(y,x)$ depends on the material properties and
geometry of the environment, both of which are subject to modeling
uncertainty and real-world variation. We therefore test whether the
performance advantage of a deployment optimized for the nominal
environment persists when the propagation model is perturbed. Using
the San Francisco AT\&T scenario, we keep both the
$\IASPA{\mathcal{P}_{\mathrm{MAX}}}$ and AT\&T
transmitter locations fixed at their nominal values and recompute their
performance under the following perturbed environments:
\begin{itemize}
    \item \textbf{Material perturbations}: the conductivity and relative permittivity of all surfaces are each scaled in turn by 
    \(\{-25\%,-10\%,+10\%,+25\%\}\), yielding eight configurations;
    
    \item \textbf{Dynamic obstacles}: $300$ street-level ``vehicles,''
    modeled as small boxes placed outside building footprints, are
    added at random using three independent random seeds;
    
    \item \textbf{Geometry changes}: one central building
    ($1{,}959$ m$^2$ footprint and $26$ m height) is removed, and a
    $30\times20\times24$ m building is added at a different location.
\end{itemize}

Only the environment is perturbed in these experiments; the
transmitter locations are not re-optimized. The comparison therefore
measures the sensitivity of the performance advantage obtained from
the nominal propagation model to mismatch between the modeled and
perturbed environments.

\begin{figure}[t]
    \centering
    \includegraphics[width=\linewidth]
    {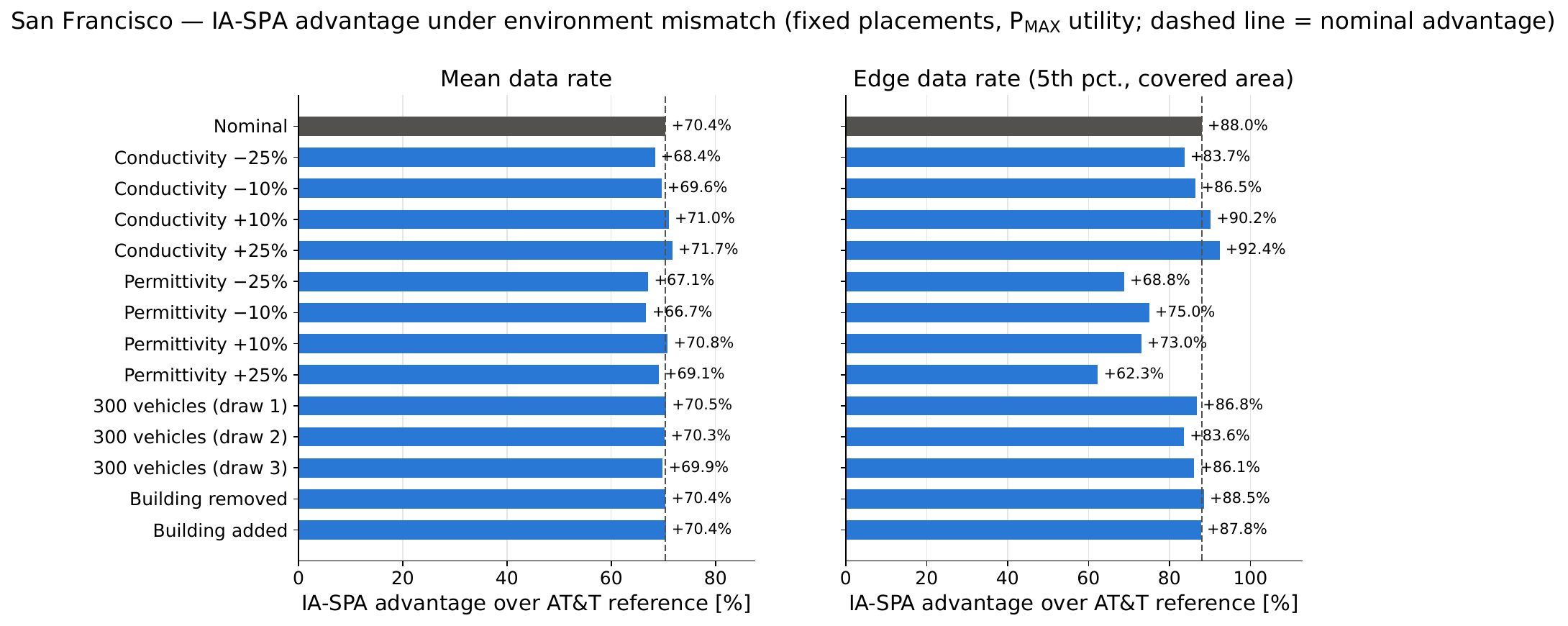}
    \caption{Percentage advantage of $\IASPA{\mathcal{P}_{\mathrm{MAX}}}$ over the fixed AT\&T
    reference in mean data rate (left) and edge rate (right) under
    perturbations of the propagation environment. The dashed line
    marks the nominal, unperturbed advantage.}
    \label{fig:exp3_advantage}
\end{figure}

Figure~\ref{fig:exp3_advantage} summarizes the results. For the nominal
environment, $\IASPA{\mathcal{P}_{\mathrm{MAX}}}$ achieves a $70.4\%$ advantage in mean rate and an
$88.0\%$ advantage in edge rate over the AT\&T reference, consistent
with Table~\ref{tab:comparison_max}. Across the thirteen perturbed
environments---eight material perturbations, three vehicle-clutter
realizations, and two geometry changes---the mean-rate advantage
remains between $66.7\%$ and $71.7\%$, while the edge-rate advantage
remains between $62.3\%$ and $92.4\%$. The smallest observed edge-rate
advantage, obtained when the relative permittivity is increased by
$25\%$, is still $62.3\%$.

In several cases, including increased conductivity and added vehicle
clutter, the relative advantage is essentially unchanged or slightly
larger than in the nominal environment. Taken together, these results
indicate that the performance advantage obtained using the nominal
propagation model is not sensitive to the material, clutter, and
geometric uncertainties considered here.

\subsection{Incremental Deployment onto Existing Networks}
\label{sec:incremental}

In practice, new transmitters are typically added to an existing
network rather than deployed in an empty environment.
$\IASPA{\mathcal{P}}$ accommodates this setting
directly by fixing the pre-existing
transmitter set $T_{\mathrm{fixed}}$ and optimizing only the locations
of the newly added transmitters, as described in
Section~\ref{sec:systemmodel}.

To test this incremental-deployment capability in already substantial
networks, we consider the full reference deployments used in the two
simulation environments and add only three transmitters. In each case,
$\IASPA{\mathcal{P}_{\mathrm{MAX}}}$ is
compared with the mean
performance obtained from ten independent trials in which three
additional transmitters are selected randomly from the feasible
candidate set.

\subsubsection{San Francisco: 16-Transmitter AT\&T + T-Mobile Network}
\label{sec:incremental_sf}

We initialize the San Francisco network with the union of the nine
AT\&T and seven T-Mobile transmitters used in the preceding
experiments, giving $16$ pre-existing transmitters. Three additional
transmitters are then selected using $\IASPA{\mathcal{P}_{\mathrm{MAX}}}$ or random placement.

\begin{figure}[t]
    \centering
    \includegraphics[width=0.75\linewidth]
    {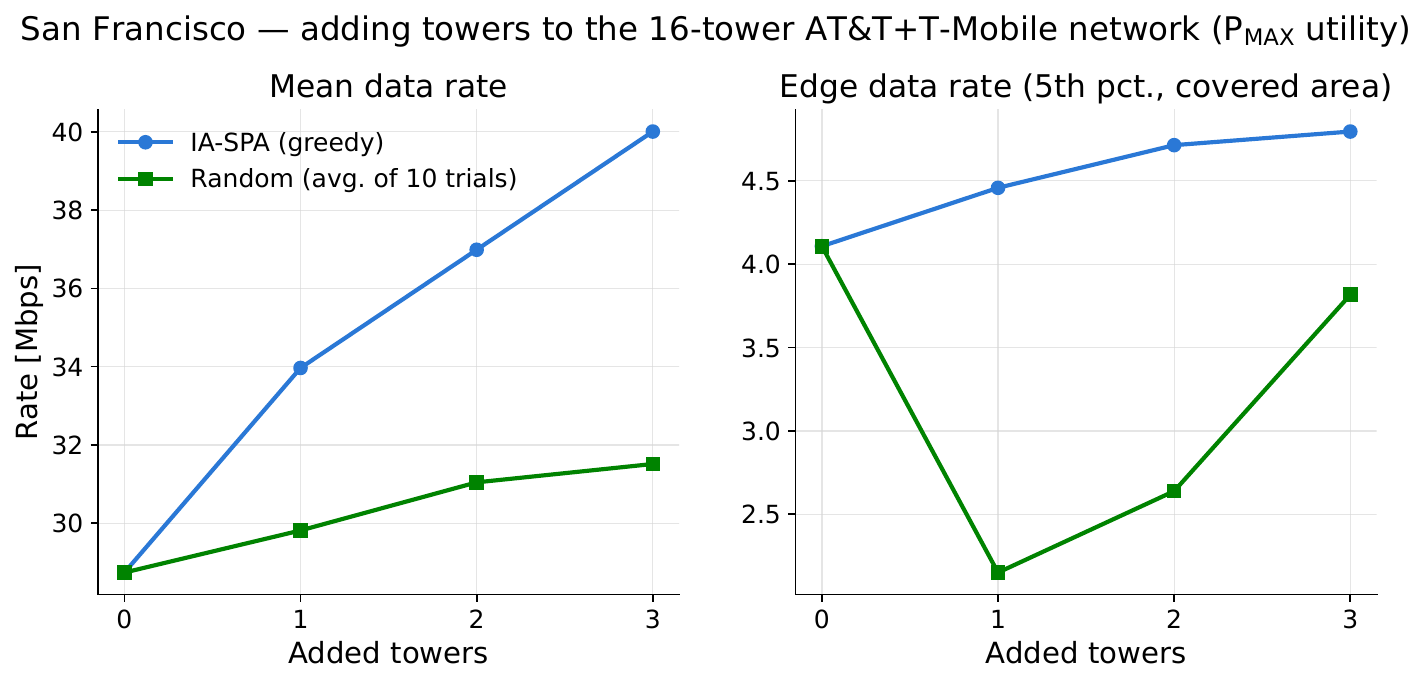}
    \caption{Incremental deployment in San Francisco starting from the
    existing 16-transmitter AT\&T+T-Mobile network. Mean data rate
    (left) and edge rate (right) are shown as transmitters are added
    using $\IASPA{\mathcal{P}_{\mathrm{MAX}}}$ or
    random placement.
    The random curve is the mean over ten independent trials.}
    \label{fig:exp5_sf}
\end{figure}

\begin{table}[t]
\centering
\renewcommand{\arraystretch}{1.2}
\caption{Incremental deployment in San Francisco after adding three
transmitters to the existing 16-transmitter AT\&T+T-Mobile network
using $\IASPA{\mathcal{P}_{\mathrm{MAX}}}$.}
\label{tab:exp5_sf}
\begin{tabular}{@{}lrrr@{}}
\toprule
\textbf{Metric}
& \textbf{Random}
& \textbf{IA-SPA}
& \textbf{Change} \\
\midrule
\textit{Data Rate [Mbps]} & & & \\
Mean Data Rate
& 31.51 & 40.01 & +26.96\% \\
Std.\ Dev.
& 25.63 & 31.88 & +24.41\% \\
Maximum Data Rate
& 167.92 & 177.85 & +5.91\% \\
Edge Rate (5\%-tile)
& 3.82 & 4.80 & +25.66\% \\
\midrule
\textit{Interference [nW]} & & & \\
Mean Interference
& 30.41 & 28.42 & -6.52\% \\
Maximum Interference
& 20500.09 & 20490.09 & -0.05\% \\
\bottomrule
\end{tabular}
\end{table}

Figure~\ref{fig:exp5_sf} and Table~\ref{tab:exp5_sf} show that the
benefit of $\IASPA{\mathcal{P}_{\mathrm{MAX}}}$ persists when transmitters are added to the existing
network. Both placement strategies start from the same 16-transmitter
baseline, whose mean rate is approximately $28.7$ Mbps. After three
additional transmitters, $\IASPA{\mathcal{P}_{\mathrm{MAX}}}$ achieves a mean rate of $40.0$ Mbps,
compared with $31.5$ Mbps for the random-placement average, a
$27.0\%$ improvement. The corresponding edge-rate improvement is
$25.7\%$.

The random-placement average also shows a substantial decrease in edge
rate after the first added transmitter before partially recovering
with subsequent additions. Thus, random densification can degrade edge
performance in an existing interference-limited network, whereas the
$\IASPA{\mathcal{P}_{\mathrm{MAX}}}$ sequence increases the edge rate at each of the three additions
considered here.

\subsubsection{Florence: 33-Transmitter Iliad Network}
\label{sec:incremental_fl}

We repeat the experiment in Florence, starting from the existing
33-transmitter Iliad deployment and adding three further transmitters.

\begin{figure}[t]
    \centering
    \includegraphics[width=0.75\linewidth]
    {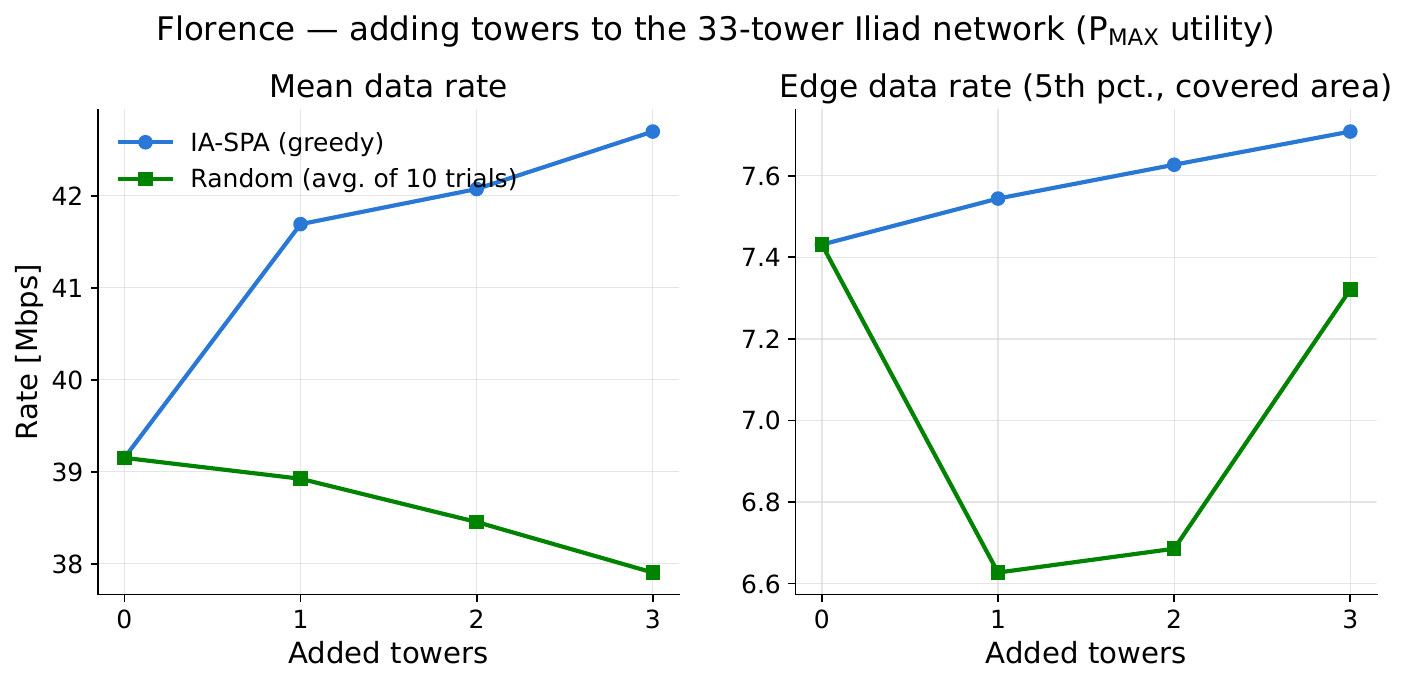}
    \caption{Incremental deployment in Florence starting from the
    existing 33-transmitter Iliad network. Mean data rate (left) and
    edge rate (right) are shown as transmitters are added using
    $\IASPA{\mathcal{P}_{\mathrm{MAX}}}$ or
    random placement. The random
    curve is the mean over ten independent trials.}
    \label{fig:exp5_fl}
\end{figure}

\begin{table}[t]
\centering
\renewcommand{\arraystretch}{1.2}
\caption{Incremental deployment in Florence after adding three
transmitters to the existing 33-transmitter Iliad network using
$\IASPA{\mathcal{P}_{\mathrm{MAX}}}$.}
\label{tab:exp5_fl}
\begin{tabular}{@{}lrrr@{}}
\toprule
\textbf{Metric}
& \textbf{Random}
& \textbf{IA-SPA}
& \textbf{Change} \\
\midrule
\textit{Data Rate [Mbps]} & & & \\
Mean Data Rate
& 37.90 & 42.70 & +12.65\% \\
Std.\ Dev.
& 25.19 & 28.89 & +14.66\% \\
Maximum Data Rate
& 194.95 & 203.53 & +4.40\% \\
Edge Rate (5\%-tile)
& 7.32 & 7.71 & +5.29\% \\
\midrule
\textit{Interference [nW]} & & & \\
Mean Interference
& 103.68 & 94.16 & -9.18\% \\
Maximum Interference
& 5086.24 & 5064.36 & -0.43\% \\
\bottomrule
\end{tabular}
\end{table}

The Florence experiment starts from a considerably denser reference
deployment. Figure~\ref{fig:exp5_fl} shows that the mean rate under
random placement decreases from approximately $39.2$ Mbps for the
33-transmitter baseline to $37.9$ Mbps after three additions. In
contrast, $\IASPA{\mathcal{P}_{\mathrm{MAX}}}$ increases the mean rate to $42.7$ Mbps, yielding a
$12.7\%$ advantage over the random-placement average. $\IASPA{\mathcal{P}_{\mathrm{MAX}}}$ also
achieves an edge rate of $7.71$ Mbps and approximately $9.2\%$ lower
mean interference than random placement.

The decrease observed under random placement is consistent with the
fact that indiscriminate densification can introduce additional
interference without a commensurate improvement in useful coverage.
In contrast, $\IASPA{\mathcal{P}_{\mathrm{MAX}}}$ selects additions according to their marginal
network-quality gain. Together with the San Francisco experiment,
these results show that the incremental formulation remains effective
when only a small number of new transmitters are added to already
substantial existing deployments.

\section{Conclusion}\label{sec:conclusion}
This paper introduced the Interference-Aware Submodular Placement Algorithm (IA-SPA), a framework for transmitter placement that combines detailed, site-specific propagation information with a mathematically tractable optimization formulation. By introducing an aggregated network-quality functional and establishing its submodularity under practical conditions, we obtain a greedy placement algorithm with a provable approximation guarantee relative to the global optimum of the surrogate objective. The framework accommodates prohibited deployment regions and existing infrastructure, making it applicable to both clean-slate and incremental network deployment.

Numerical experiments using high-fidelity ray tracing in San Francisco
and Florence demonstrate substantial improvements in mean and edge-user
data rates relative to existing operator deployments, while also
reducing interference. These gains persist under multiple exclusionary
zones, incremental densification of already dense networks, small-scale
fading, and perturbations of material properties, obstacle geometry,
and street-level clutter. In particular, the robustness experiments
indicate that the performance advantage obtained using the nominal
propagation model is not sensitive to the forms of model uncertainty
considered here. The computational analysis further shows linear
dependence on the number of candidate sites for fixed receiver
discretization and deployment size, while the measured runtimes
demonstrate practical feasibility at the city scale considered here.
Taken together, these results demonstrate that detailed site-specific
propagation models can be incorporated into a computationally tractable
and theoretically grounded framework for transmitter placement. 

Several extensions are natural. Heterogeneous, spatially varying deployment costs could account for real-estate and permitting costs as well as different types of base stations. Placement could also be coupled with operational variables such as transmit power and scheduling in response to time-varying demand. Finally, evaluating the framework across a broader range of urban environments and propagation regimes would help establish its applicability beyond the two cities considered here.

\bibliographystyle{IEEEtran}
\bibliography{bibliography, Andrews}

\end{document}